\documentclass[letterpaper,twocolumn,prx,aps,showpacs,superscriptaddress,floatfix]{revtex4-2}

\usepackage{amssymb, amsmath, bbm, amsthm, changes, graphicx, dsfont, verbatim, appendix, physics} 

\usepackage[hyperindex, breaklinks]{hyperref}
\hypersetup{
     colorlinks=true,        
     citecolor=blue,    
     filecolor=blue,      		
     urlcolor=blue,           	
    runcolor=cyan,
}

\usepackage{enumitem}   
\usepackage{xcolor}
\usepackage{orcidlink}

\def\Z2{\mathbb Z_2}

\def\r{\textrm{r}}
\def\ir{\textrm{ir}}

\def\I{\mathrm i}
\def\unity{\mathbbm 1}

\def\pr{{\rm prob}}

\theoremstyle{definition}
\newtheorem{Lemma}{Lemma}
\newtheorem{Definition}[Lemma]{Definition}

\newtheorem{Theorem}[Lemma]{Theorem}
\newtheorem{Corollary}[Lemma]{Corollary}

\begin{document}
\title{Absence of localization in two-dimensional Clifford circuits}

\author{Tom Farshi\,\orcidlink{0000-0003-0504-7142}}
\affiliation{Department of Computer Science, University College London, 
United Kingdom}
\affiliation{Department of Physics and Astronomy, University College London, 
United Kingdom}
\affiliation{School of Mathematics, University of Bristol, United Kingdom}

\author{Jonas Richter\,\orcidlink{0000-0003-2184-5275}}
\affiliation{Department of Physics and Astronomy, University College London, 
United Kingdom}
\affiliation{Department of Physics, Stanford University, USA}
\affiliation{Institut f\"ur Theoretische Physik, Leibniz Universit\"at 
Hannover, Germany}

\author{Daniele Toniolo\,\orcidlink{0000-0003-2517-0770}}
\affiliation{Department of Computer Science, University College London, 
United Kingdom}
\affiliation{Department of Physics and Astronomy, University College London, 
United Kingdom}

\author{Arijeet Pal\,\orcidlink{0000-0001-8540-0748}}
\affiliation{Department of Physics and Astronomy, University College London, 
United Kingdom}

\author{Lluis Masanes\,\orcidlink{0000-0002-1476-2327}} 
\email[]{l.masanes@ucl.ac.uk}
\affiliation{Department of Computer Science, University College London, United 
Kingdom}
\affiliation{London Centre for Nanotechnology, University College London, 
United Kingdom}

\begin{abstract}
We analyze a Floquet circuit with random Clifford gates in one and two spatial dimensions. 
By using random graphs and methods from percolation theory, 
we prove in the two dimensional setting that some local operators grow at 
ballistic rate, which implies the 
absence of localization. 
In contrast, the one-dimensional model displays a 
strong form of localization characterized by the emergence of left and 
right-blocking walls in random locations.
We provide additional insights by complementing our 
analytical results with numerical 
simulations of operator spreading and entanglement growth, 
which show the absence (presence) of localization in 
 two-dimension (one-dimension). 
Furthermore, we unveil that the spectral form factor of the Floquet unitary in two-dimensional 
circuits behaves like that of quasi-free fermions with chaotic single 
particle dynamics, 
with an exponential ramp that persists till times scaling 
linearly with the size of the system. Our work sheds light on 
the nature of disordered, Floquet Clifford dynamics and its relationship to fully chaotic quantum 
dynamics. 
\end{abstract}

\maketitle

\section{Introduction}

Understanding the dynamics of quantum many-body systems far from equilibrium is 
of fundamental importance for preparing and controlling quantum states of matter 
\cite{Polkovnikov_2011}. The universal dynamical behavior provide signatures of 
novel quantum phase of matter and their underlying patterns of quantum 
information. Studying the dynamics of quantum many-body systems is notoriously 
challenging due to the exponential growth of the 
Hilbert space. In recent years, simulating dynamics in a quantum circuit  
architecture has opened new directions for probing quantum chaos, hydrodynamics, 
and non-equilibrium phases of matter \cite{Nahum_2017, Nahum_2018, 
von_Keyserlingk_2018, Rakovszky_2018, Khemani_2018a, Zhou_2019, Brandao_2019, 
Potter2021}. They are particularly suitable for quantum simulation on noisy 
intermediate-scale quantum (NISQ) devices and have been studied on the 
state-of-the-art physical platforms. Novel quantum many-body phenomena can be 
realised in these experiments providing a test bed for the theoretical 
ideas with potential applications in protecting and processing quantum 
information.

The dynamics in circuit models are encoded in the form of k-local unitary gates  
acting on qubits. The geometry of the gates and their symmetries provide access 
to a wide range of model phenomena which are analytically tractable and can be 
approximated efficiently. Due to their tunability and control, circuit models 
provide minimal models for complex quantum phenomena,  including systems with 
kinetic constraints 
\cite{Moudgalya_2021, Gopalakrishnan_2018}, dual-unitary structure 
\cite{Bertini_2019_2, Claeys_2021}, periodic dynamics \cite{Farshi_2022}, 
or long-range interactions \cite{Richter_2022}. Notably, quantum circuits with  
random gates have 
provided a powerful framework to strive for a quantum computational advantage 
\cite{Arute_2019} as well as to study quantum information scrambling on 
existing quantum hardware \cite{Mi_2021}.

The vast majority of many-body quantum systems are 
ergodic and relax to thermal equilibrium under unitary time evolution 
\cite{D'Alessio_review_2016, Gogolin_2016}.
Exceptions to this generic thermalizing behavior 
include integrable models \cite{Vidmar_2016}, which possess an extensive set of 
conservation laws, as well as models exhibiting localization 
\cite{Nandkishore_2015, Abanin_2019}. Localization can arise in systems with 
sufficiently strong disorder, and one of its signatures is the slow growth of 
any initially-local operator when evolving in the Heisenberg picture. 
In a quantum circuit, localization in terms of operator growth can be of  
single-particle type, where the support of any initially local operator 
remains confined in a finite region for all times \cite{Stolz_2011, 
Aizenman_2015}, as realized for instance in lattice models of non-interacting 
fermions with random on-site potential \cite{Anderson_1958} and for the corresponding mapping to spin-systems \cite{Sims_Stolz_2012, Nach_2016}.   
For a typical realization of our circuit in one dimension, this definition takes into account the 
appearance of walls that block the spread of any operator, like illustrated in 
Fig.\ \ref{fig:lightcone}~(a). 
On the other hand, 
many-body localization exhibits the growth of support of a local operator as 
the logarithm of time 
\cite{Das_Sarma_2017}. In particular the log-like growth of entanglement 
entropy in time \cite{Znidaric_2008, Bardarson_2012,Serbyn_2013} tells these 
systems apart from the Anderson-type localized as defined below 
in definition \ref{def_Anderson_type}. 
In one-dimensional systems, a many-body localized phase with an 
extensive set of exponentially localized integrals of motion might exist at 
sufficiently strong disorder \cite{Imbrie_2016}. However, the 
asymptotic existence of this MBL phase in the thermodynamic limit is still 
under active debate as localization has recently been found to be unstable  
even at rather large values of disorder \cite{De_Roeck_2017_1, _untajs_2020, 
Prosen_chal_2020, Sierant_2020, Morningstar_2022}.
Moreover, in higher dimensions, analytical arguments and 
numerical calculations suggest that MBL is 
unstable \cite{deRoeck_Imbrie_stable_2017, Wahl_2018, 
Sels_2021, Kiefer_Emmanouilidis_2020, Roop_2022, Richter_2022_2}.

Disorder in quantum circuit models can be introduced in space and time where  
the gates are chosen at random. Random unitary dynamics without any symmetries 
or constraints lead to complete mixing \cite{Nahum_2018, von_Keyserlingk_2018, 
D_Alessio_2014, Lazarides_2014} 
while certain time-periodic circuits can 
exhibit non-ergodic dynamics \cite{Farshi_2022, Sunderhauf_2018}. 
There are also various forms of kinetically 
constrained random circuits, akin to fractonic models, which exhibit 
localization \cite{Khemani_2020, Pai_2019}. In this work we 
shed light on the role of 
dimensionality on 
localization in random Floquet Clifford circuits. 
Floquet circuits have 
been extensively studied 
\cite{Prosen_2018_1,Prosen_2018_2,Bertini_2019, Chalker_2018, 
Chalker_PRX_2018, Sunderhauf_2018} and have provided key insights into the 
properties of periodically-driven quantum systems \cite{Bukov_review_2015}. They 
have been essential to 
rigorously demonstrate the occurrence of chaos and random-matrix behavior in 
isolated quantum systems \cite{Prosen_2018_1, Chalker_2018} and, in combination with disorder 
and many-body 
localization, they can host exotic non-equilibrium phases of matter with no 
equilibrium counterpart \cite{Khemani_2016, Else_2016}. Moreover, certain 
circuits with dual-unitary 
structure allow for a controlled tuning between ergodic and 
non-ergodic dynamics \cite{Claeys_2021, Bertini_2019_2}. In this work, we analytically and numerically study 
absence of localization in two-dimensional Clifford circuits and characterise 
the integrable nature of chaos using measures of operator growth and spectral 
form factor. Moreover we numerically contrast the two-dimensional against the one-dimensional case showing for the latter localization at the dynamical level. 

From a quantum information perspective, the Clifford group plays a key role in
fault-tolerant quantum computing and randomized benchmarking 
\cite{Knill_2008, Magesan_2011, Onorati_2019}. Recently, random circuits consisting 
of Clifford gates have provided useful insights into 
quantum many-body physics \cite{Nahum_2017, Li_2018, Farshi_2022, Ludwig_2021, Lunt_2021} due to their 
efficient 
simulability on classical computers despite the generation of extensive 
amounts of entanglement. This efficient simulability of Clifford dynamics can 
also be understood in terms of a phase space 
representation analogous to that of quasi-free bosons and fermions, with the 
dimension of the phase space being exponentially smaller than the Hilbert 
space, 
see appendix A of \cite{Farshi_2022} and  \cite{Gottesman_1998, Aaronson_2004}. 
Yet, Clifford circuits 
form unitary designs \cite{Harrow_2009, 
Webb_2016, Zhu_2017} such that circuit averages of certain relevant quantities 
can exactly reproduce Haar averages over the full unitary group. 
Despite their simplicity, Clifford circuits can thus prove useful to gain 
insights into some aspects of the dynamics of more generic quantum systems, 
including the buildup of out-of-time-ordered correlators or the growth of 
entanglement \cite{Nahum_2018, von_Keyserlingk_2018}.

In addition, as we will demonstrate, the Clifford circuits studied in this paper 
are to some extent 
tractable analytically by a suitable mapping to directed graphs. It is known 
that random time-dependent (i.e.~annealed disorder) cellular automata can be 
analysed with directed-percolation theory \cite{Derrida_1986, Liu_2021}, and 
that Clifford circuits can be represented as cellular automata. However, the 
circuits considered by us are not random in 
time, they have quench disorder, and hence, they cannot in general be 
solved with directed graphs. Nevertheless, as we 
will show below, in order to probe their behaviour, we only 
need to analyse the dynamics at the edge of the light-cone. Moreover, 
at the light-cone edge, spatial disorder is equivalent to time disorder, 
producing an effectively annealed dynamics, which can be analysed with 
percolation theory.

The hybrid nature of Clifford circuits between integrable and chaotic 
systems also reflects itself in the emergence of ergodicity and 
localization. On one hand, it was shown in Ref. \cite{Farshi_2022}, 
that Floquet Clifford circuits exhibit Anderson-type localization in  
one 
dimension (1D), cf.\ Fig.\ \ref{fig:lightcone}~(a) and 
definition \ref{def_Anderson_type}. On the other hand, we 
prove here, and numerically show cf.\ Fig.\ 
\ref{fig:lightcone}~(b),  as 
a main result, that in two dimensions (2D) some operators always grow at 
ballistic rate such that the model 
does not localize, despite having strong disorder, contrasting 
the phenomenon in one dimension. 
We also provide numerical evidence that the ballistic growth happens for almost 
all local operators.
Moreover, the absence of localization in our time-periodic 2D circuit 
model differs from the behavior of disordered free-fermion 
models which show Anderson localization for arbitrarily weak disorder in 2D 
\cite{Abrahams_1979, Lee_1985, Evers_2008}.

We elucidate further aspects of the Floquet Clifford dynamics by complementing 
our analytical 
results with numerical simulations, where we 
demonstrate that operator spreading is exponentially suppressed in 
1D, reminiscent of the exponential dynamical-localization of the wave function 
in Anderson insulators. On the contrary, we find that operator 
spreading in 2D occurs ballistically with light-speed velocity and that 
the interior of the light-cone becomes fully scrambled and featureless. These 
findings are substantiated by the temporal buildup of entanglement, 
which is bounded and system-size independent in 1D, whereas 
it grows linearly and saturates towards extensive values in 2D. Finally, we also 
study the spectral form factor (SFF) of the Floquet Clifford circuit,
which is a key quantity to diagnose the emergence of quantum 
chaos. 
We show that the SFF is similar to that of free fermions whose associated 
single-particle dynamics is chaotic.
Specifically, we find that, in the 2D model, the SFF exhibits an 
exponential 
ramp at early times that persists up to a time that scales linearly with the 
size of the system,
suggesting that ergodicity in 
the case of Clifford dynamics should be understood with respect to the 
exponentially smaller phase space.
\begin{figure}[tb]
  \center
  \includegraphics[width=\columnwidth]{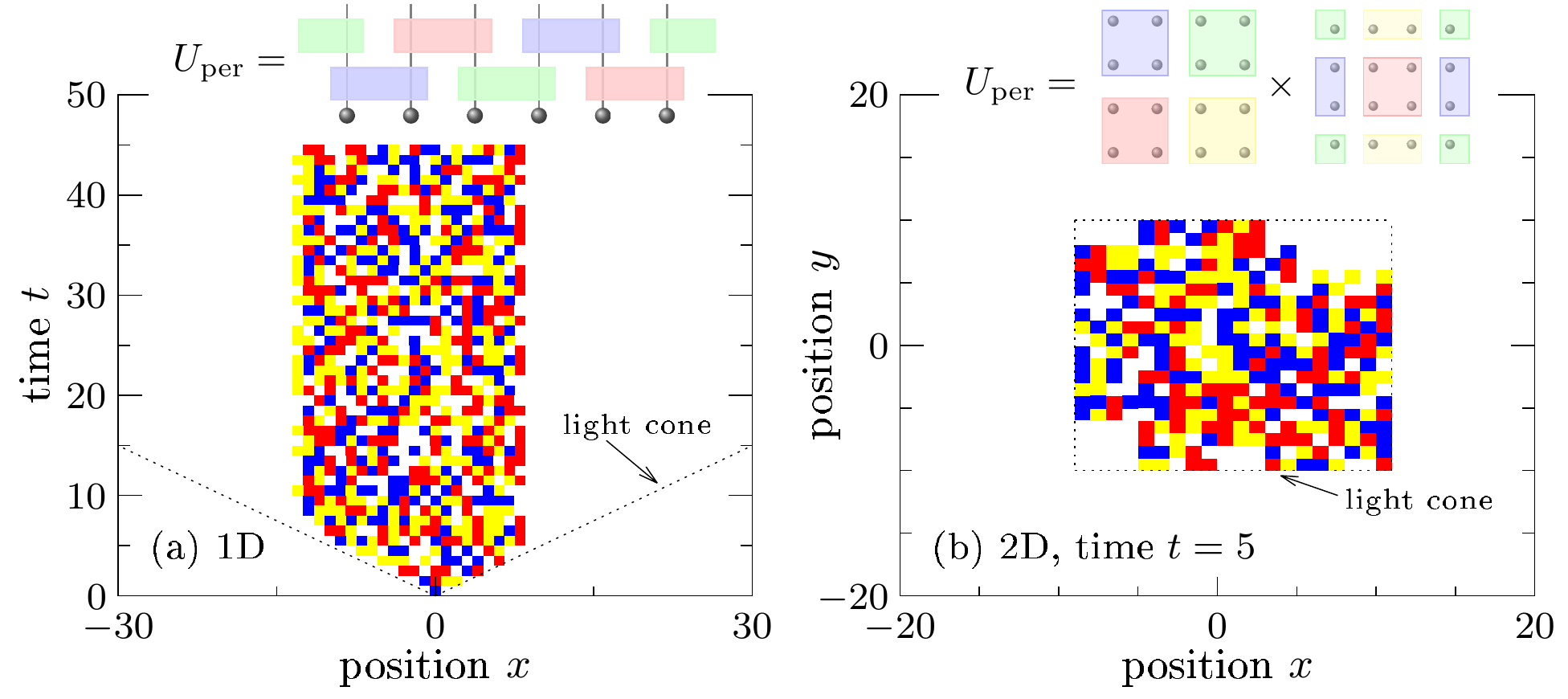}  
  \caption{{\bf Localization and its absence in Floquet 
Clifford circuits.} 
A local operator $\sigma_x$ acting non-trivially only on 
a single site evolves according to a time-periodic Clifford circuit in (a) 
one dimension and (b) 
two dimensions. In 1D, the Floquet unitary $U_\text{per}$ consists of a 
brickwork pattern of 2-qubit gates, while in 2D it consists of two layers of 
4-qubit gates, cf.\ Fig.\ \ref{fig:Localization_2D_Absence_ModelTimeEvolution}. 
The color plot encodes the local matrices of the time-evolved operator string 
$U(t)\sigma_x U^\dagger(t)$ according to $\unity$ 
(white), $\sigma_x$ (blue), $\sigma_y$ (red), $\sigma_z$ (yellow). Data is 
shown for a single random realization of the Clifford gates. For better 
visualization of the 2D data, we focus in 
panel (b) on a single point in time, $t = 5$, 
where the light cone boundary (dotted square) is of size $20 \times 20$, cf.\ 
definition \ref{Def:lightcone}.
As proven in Ref.\ \cite{Farshi_2022}, the 1D circuit exhibits localization 
due to the emergence of left- and right-sided walls that confine the evolution 
at all times. This confinement affects all (not necessarily Pauli) operators 
with support between the two walls. Inside the confined region the evolution 
looks ergodic. In contrast, in 2D, localization is absent as can be seen from 
the fact that parts of the operator grow with light-cone speed. 
}
\label{fig:lightcone}
\end{figure}

Our work provides a 
comprehensive picture of localization and chaos in disorderd, Floquet Clifford 
quantum-circuits, in terms of directed percolation, at the light-cone, and 
information spreading in classical cellular automata. Understanding of 
non-equilibrium quantum many-body states in Clifford circuits provides an 
important starting point for studying the phenomena in generic conditions. It 
can serve as a useful benchmark for identifying quantum effects of simulation in 
noisy quantum devices. The relevance of Clifford circuits for quantum 
error-correction can also provide valuable insights into combining 
error-correction protocols with quantum simulation.

The rest of this paper is structured as follows.
Section \ref{sec:model} contains a description of the model, including the 
definition of Clifford unitaries, which are the building blocks of this quantum circuit. 
Section \ref{sec:results} contains the main result of our work, theorem 
\ref{thm:Localization_2D_Absence_AlternativeModel_AbsenceResult}, together with 
a discussion of its significance, as well as its numerical 
demonstration in Sec.\ \ref{Sec::Numerics}. 
In Sec.\ \ref{Sec::SFF}, we numerically study the spectral 
form factor of Floquet Clifford unitaries which 
supports our observation of localization in 1D, see also \cite{Farshi_2022}, 
and ergodic 
dynamics in 2D. We also provide a rigorous lower bound on the time-averaged 
spectral form factor.
Section \ref{sec:proof} contains the proof of theorem 
\ref{thm:Localization_2D_Absence_AlternativeModel_AbsenceResult}, which uses 
random graphs and methods reminiscent of those of percolation theory. 
 Eventually, in Sec.\  
\ref{sec:int/chao}, we provide a discussion about the fact that 
Clifford dynamics appears to share properties of integrable systems and chaotic 
systems, Sec.\ \ref{sec:conclusions} summarizes the results of this work and 
provides some outlook. Appendix \ref{App} provides the relation among 
localization length and the average position of a blocking wall in the dynamics 
of the one-dimension model.

\section{Results}

\subsection{Description of the model}
\label{sec:model}

Consider an infinite two-dimensional square lattice with sites labeled by 
$(x,y)\in \mathbb Z^2$. In each site there is a spin-1/2 particle, or qubit, 
which has Hilbert space $\mathbb C^2$. The dynamics of the system is 
time-periodic, with the period being one unit of time. Hence, at integer times 
$t\in \mathbb Z$, the evolution operator can be written as $U(t) = 
(U_{\text{per}})^t$. The unitary $U_{\text{per}}$ that describes the dynamics of a single period has the form
\begin{equation}
  \label{eq:Localization_2D_Absence_Depth2_Uperiod}
  U_{\text{per} } = \left( \bigotimes_{ x,y \text{ odd}} W_{(x,y)} \right)  
\left( \bigotimes_{ x,y\text{ even}} W_{(x,y)} \right) ,
\end{equation}
where, for any $(x,y)\in \mathbb Z^2$, the local unitary $W_{(x,y)}$ acts on  
the four sites $(x,y)$, $(x+1,y)$, $(x,y+1)$ and $(x+1,y+1)$. In 
\eqref{eq:Localization_2D_Absence_Depth2_Uperiod}, $ ( x,y \text{ odd} ) $ means 
that both $ x $ and $ y $ are odd, analogously $ ( x,y \text{ even} ) $ means 
that both $ x $ and $ y $ are even. In the dynamics described by 
\eqref{eq:Localization_2D_Absence_Depth2_Uperiod} each time period decomposes 
into two half time-steps, which are illustrated in Fig.\ 
\ref{fig:Localization_2D_Absence_ModelTimeEvolution}. 
This quantum circuit with local interactions is an example of 
a quantum cellular  
automaton \cite{Farrelly_2019, Gutschow_2010, 
Schlingemann_2008, Gopalakrishnan_2018}.
The evolution operator $U(t)$ can also be generated by a time-periodic  
Hamiltonian $H(t)$ with local interactions
\begin{equation}
  U(t) = \mathcal T e^{-\I \int_0^t d\tau H(\tau)}\ ,
\end{equation}
where $\mathcal T$ stands for time-ordered exponential.
This type of dynamics is called Floquet dynamics \cite{Bukov_review_2015}.
\begin{figure}[tb]
\center
\includegraphics[width=0.85\columnwidth]{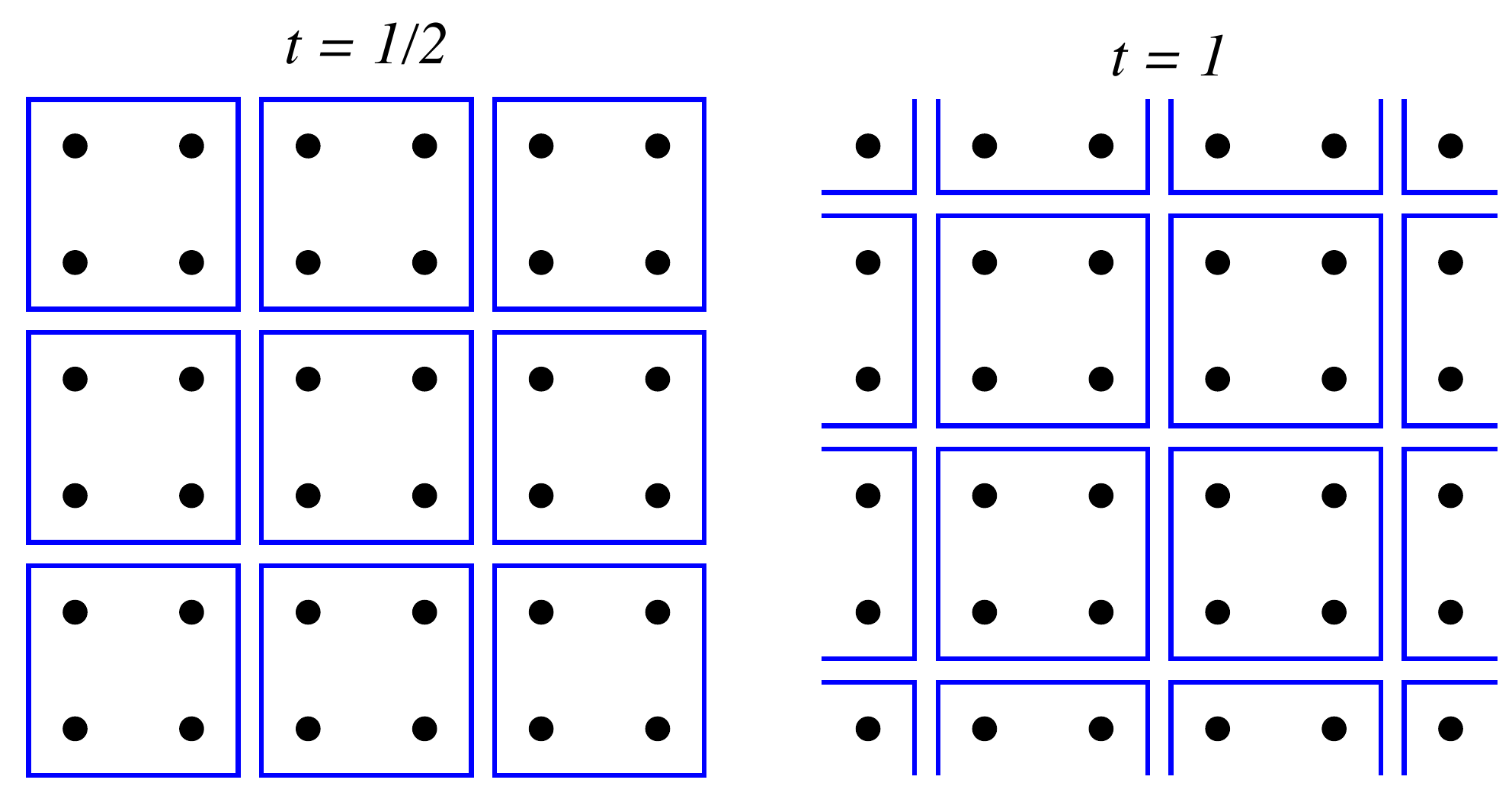}  
\caption{\textbf{Local dynamics.} This figure represents the two layers of 
local 
unitaries generating the dynamics of a time period $U_{\text{per} }$. Each black 
dot represents a qubit and each blue square represents a four-qubit unitary 
$W_{(x,y)}$.}
\label{fig:Localization_2D_Absence_ModelTimeEvolution}
\end{figure}

The disorder of the system is represented by the fact that each local unitary  
$W_{(x,y)}$ is sampled independently from the uniform distribution over the 
4-qubit Clifford group. A unitary $W$ is Clifford if it maps any product of 
Pauli sigma matrices to another product of Pauli sigma matrices 
\cite{Gottesman_1998, nielsen_chuang_2010}. An example of Clifford unitary $W$ 
acting on $\mathbb C^2 
\otimes \mathbb C^2$ is
\begin{align}
  W(\sigma_x \otimes \unity) W^\dagger
  &= \sigma_z \otimes \sigma_x\ , \label{Eq::Cliff1} \\
  W(\sigma_z \otimes \unity) W^\dagger
  &= \unity \otimes \sigma_z\ , 
  \\
  W(\unity \otimes \sigma_x) W^\dagger
  &= \sigma_z \otimes \unity\ , \\
  W(\unity \otimes \sigma_z) W^\dagger
  &= \sigma_x \otimes \sigma_z\ .
\end{align}
It can be seen that this four identities fully specify $W$ up to a global phase. 
Clearly, having random Clifford unitaries in neighboring unit 
cells is a somewhat different type of disorder than the one usually considered 
in the study of localization (e.g., random on-site magnetic fields in 
spin-chain models \cite{Nandkishore_2015, Abanin_2019}). For instance, apart 
from the fact that the unitaries are 
drawn independently, it is not immediately obvious how to quantify the strength 
of disorder in the random circuit model.

In summary, we have an ensemble of models of which we are going to study the  
typical behaviour. Note that, while the statistics defining the model is 
translation invariant, typical instances are not. This approach is standard in 
the study of disorder and localization \cite{Stolz_2011}. We also have to 
mention that the Floquet Clifford circuit described in 
Eq.\ \eqref{eq:Localization_2D_Absence_Depth2_Uperiod} and Fig.\ 
\ref{fig:Localization_2D_Absence_ModelTimeEvolution} is a 
straightforward 2D version of the 
1D model considered in \cite{Farshi_2022}. Remarkably, in 
Ref.\ \cite{Farshi_2022} it was proven that the 1D Floquet model 
with 2-qubit Clifford gates acting in a brickwork pattern supports 
Anderson-type localization. More specifically, initially-local operators remain 
confined to a finite region in space during the time evolution due to the 
emergence of left and right-sided walls, cf.\ 
Fig.\ \ref{fig:lightcone}~(a), that confine the operator spreading at 
all times, see also \cite{Chandran_2015}. 
Here, we will proof that this kind of localization is 
absent in the 2D model, cf.\ Fig.\ \ref{fig:lightcone}~(b). In this context, 
let us note that this absence of localization in 2D is in contrast to the 
results of Ref.\ \cite{Chandran_2015}, which studied similar Floquet Clifford 
circuits hosting a localization-delocalization transition in 2D. However, while 
Ref.\ \cite{Chandran_2015} considered a drastically reduced gate set (only two 
different two-qubit Clifford elements), our circuits are more general with 
gates being drawn at random from the full Clifford group. 

One difference between the 1D [Fig.\ \ref{fig:lightcone}~(a)] and the 
2D model [Fig.\ \ref{fig:lightcone}~(b)] is that 
in the former case the gates are sampled from the 2-qubit Clifford 
group ${\cal C}_2$, while in the latter case they are sampled from the 4-qubit 
Clifford group ${\cal C}_4$. Both 
contain a finite number of distinct elements, but their dimensions differ 
substantially, i.e., $|{\cal C}_2| = 11520$ and $|{\cal C}_4| \approx 1.2\times 
10^{13}$
\cite{Koenig_2014}.  
\begin{figure}[tb]
\center
\includegraphics[width=0.9\columnwidth]{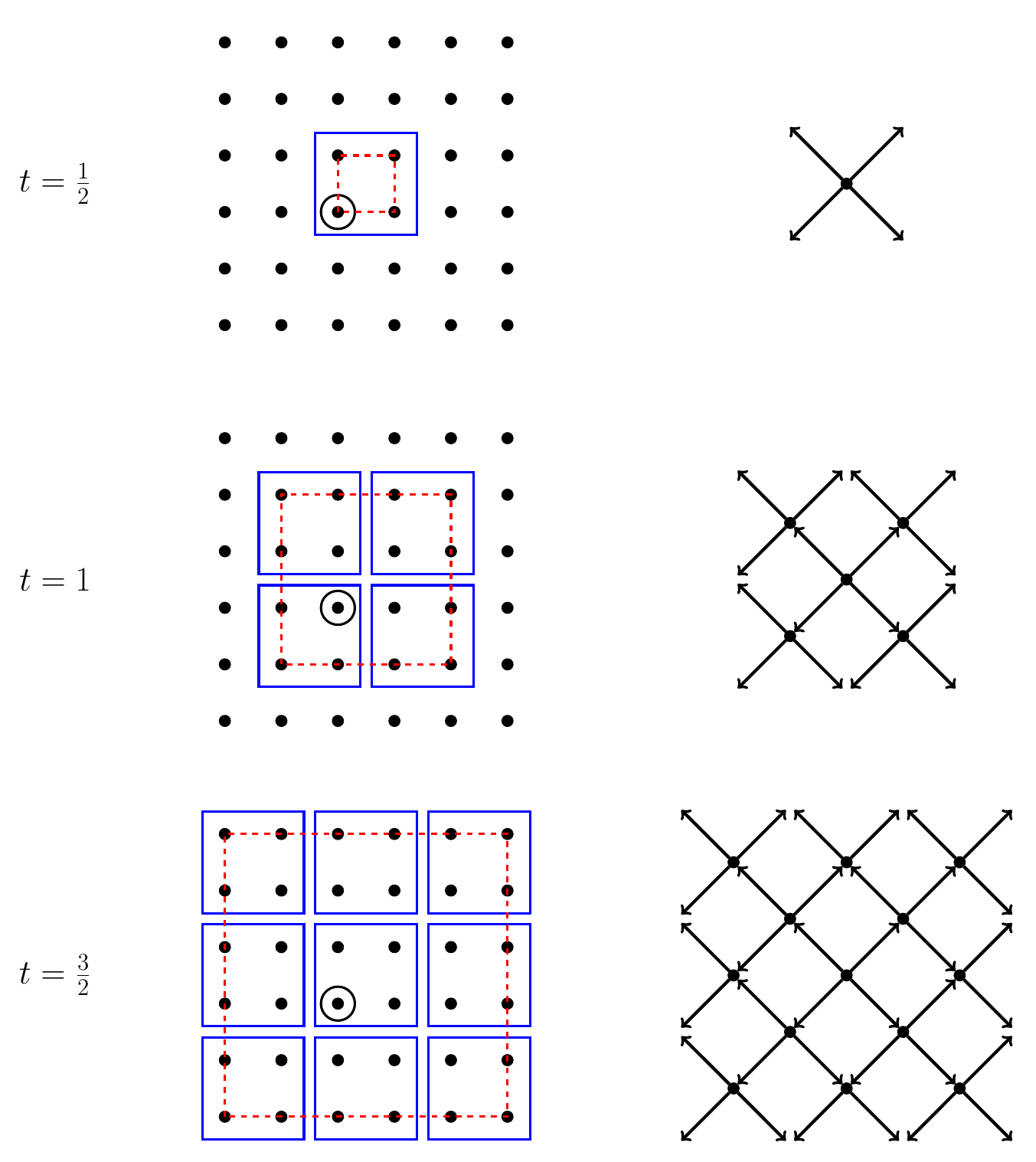}
\caption{\textbf{Light-cone boundary growth.}
  The left column shows the light-cone and its boundary at times  $t \in  
\{1/2, 
1, 3/2\}$ of an operator acting on the encircled site at $t=0$.
  Black dots represent sites, blue squares represent four-qubit unitaries and 
the red dashed line is the light-cone boundary.
  The right column shows the piece of the directed graph $G$ which describes 
  the growth history of the boundary up to the corresponding time $t$. 
  Each vertex in $G$ represents a unitary and each arrow represents a qubit at 
  a particular time $t$. The arrows which do not point to any vertex represent 
the qubits of the boundary at time $t$.}
\label{fig:Localization_2D_Absence_AlternativeModel_BoundaryGrowth}
\end{figure}

In the following, our analytical arguments 
will be focused on the 2D model 
\eqref{eq:Localization_2D_Absence_Depth2_Uperiod} and we refer the interested 
reader 
to Ref.\ \cite{Farshi_2022} for the proof of localization in 1D. 
However, in order to provide a concrete comparison of their properties, we 
will present 
numerical results of operator spreading and entanglement growth (Sec.\ 
\ref{Sec::Numerics}), as well as simulations of the spectral form factor (Sec.\ 
\ref{Sec::SFF}), for both 1D and 2D Floquet Clifford 
circuits.  

\subsection{Absence of localization}
\label{sec:results}

Many-body lattice systems display localization when the Lieb-Robinson velocity vanishes. Next we provide a definition of Anderson-type localization based on \cite{Sims_Stolz_2012}. 
\begin{Definition} \label{def_Anderson_type}
Let $O^{H}(t) = e^{iHt} O e^{-iHt} $ be the evolution in the Heisenberg picture of a local observable $O$ with support at the origin of the lattice, and let $ O^{H}_l(t) $ be the restriction of $ O^{H}(t) $ onto the ball of radius $ l>0 $ around the origin.
A system displays \textbf{Anderson-type localization} if there are parameters $ c, \mu > 0 $ such that
\begin{align}\label{eq_def_loc}
 \underset{H}{\mathbb{E}} 
 \sup_t \| O^{H}_l(t) - O^{H}(t) \|_\infty \le c \, e^{-\frac{l}{\mu}}\ ,  
\end{align}
for all $l>0$.
(The operator norm $\|A\|_\infty$ is the spectral radius of $A$, that for finite systems is the largest singular value of $A$.)
\end{Definition}
Note that the left-hand side of \eqref{eq_def_loc} involves an average over all realisations of the dynamics $H$, which usually correspond to  different realisations of a random potential.
In the case of QCAs or Floquet systems, the average is taken over the single-period unitary $U_{\text{per} }$ characterising the dynamics $O(t) = U_{\text{per} }^{-t} H U_{\text{per} }^{t}$.
In particular, in quantum circuits, the average runs over all realizations of the circuit. 
The above definition implies that, if a system has a dynamics such that a local operator grows at linear rate (ballistically) in some direction, then we can conclude that there is no localization of any type. With this in mind we make the following definitions.

\begin{Definition}
\label{Def:lightcone}
  The \textbf{light-cone} 
  is the largest support that a single-site operator at $t=0$ could have at each 
 future time $t$ when we maximize over instances of the dynamics (see figure 
\ref{fig:Localization_2D_Absence_AlternativeModel_BoundaryGrowth}).
  At each time $t$, the \textbf{boundary} of the light-cone consists of the  
outer qubits contained in the light-cone, which form a square of side $4t$.
\end{Definition}

\noindent
The approximate shape of the light-cone is a four-sided pyramid, the apex of 
which is the site where the initial operator is supported. The surface of this 
pyramid is the boundary of the light-cone. In general, an operator may not have 
full support inside the light-cone, because it acts as the 
identity in some sites at particular times (see figure 
\ref{fig:Localization_2D_Absence_DefinitionExample}). However, despite not 
having full support, an operator may grow at maximal speed, as defined next.

\begin{Definition}
\label{def:Localization_2D_Absence_AbsenceLocalization}
A model has \textbf{light-speed operator growth} if there is a single-site 
operator whose evolution has non-trivial support (is non-identity) on the 
boundary of the light-cone at all future times $t=1,2,3,\ldots$ (see figure 
\ref{fig:Localization_2D_Absence_DefinitionExample}).
\end{Definition}

\begin{figure}[tb]
\center
\includegraphics[width=0.45\columnwidth]{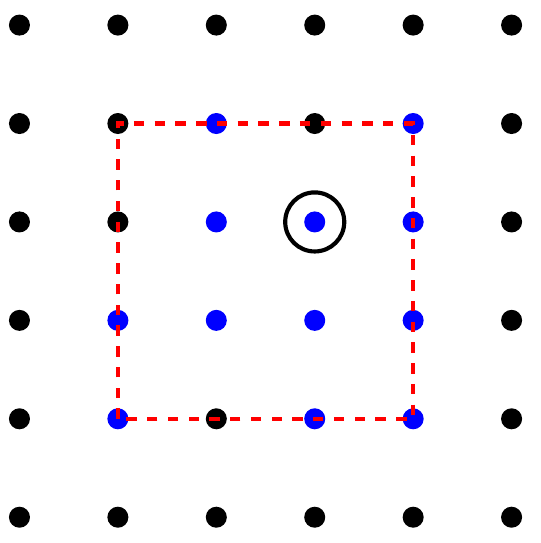}
\caption{\textbf{Light-speed 
growth.} In this example, an operator located at the circled site at time $t=0$ 
evolves into an operator with support on the blue sites at time $t=1$. Since 
some of this blue sites are at the light-cone boundary (red dashed line) this 
operator enjoys light-speed growth (at least up to time $t=1$).
}
\label{fig:Localization_2D_Absence_DefinitionExample}
\end{figure}

\noindent
Light-speed operator growth implies absence of localization.
However, not having light-speed operator growth does not imply localization. 
For example, ballistic growth at a rate smaller than unity is not localization. 
Also, diffusive dynamics, where the diameter of the support of an operator  
grows like the square root of time is not localization.
In what follows we introduce the main result of this work, which establishes 
the absence of localization in our model.

\begin{Theorem}
\label{thm:Localization_2D_Absence_AlternativeModel_AbsenceResult}
  Consider a Pauli operator which differs from the identity only on a single site at $t=0$.
  With probability at least $0.44$, the time evolution of this operator 
  generated by the random dynamics 
\eqref{eq:Localization_2D_Absence_Depth2_Uperiod} is non-identity on some sites 
of the boundary of the light-cone at all times $t \in \{1/2 , 1 , 3/2 , 2 , 5/2 
, \ldots \}$.
\end{Theorem}
\noindent
That is, with probability 0.44, a particular Pauli operator at a particular 
site enjoys light-speed growth. Our numerical simulations strongly suggests that this happens for a fraction much larger than 0.44.
Non-Pauli operators are more prone to light-speed 
growth, because they have more than one term when expressed in the Pauli basis, 
and hence more likelihood that at least one term has light-speed growth.

Before turning to the formal proof of theorem 
\ref{thm:Localization_2D_Absence_AlternativeModel_AbsenceResult} in Sec.\ 
\ref{sec:proof}, we now provide numerical support for the presence 
(absence) of localization in 1D (2D) random Floquet 
Clifford circuits. In particular, let us note that while some aspects of the 
Clifford 
dynamics can be treated analytically (cf.\ Sec.\ \ref{sec:proof}), others are 
accessible only by numerical means.  

\subsection{Numerical analysis of 
Floquet Clifford circuits}\label{Sec::Numerics}

\subsubsection{Simulating Clifford circuits}

Clifford circuits can be efficiently simulated on classical 
computers by exploiting the stabilizer formalism \cite{Gottesman_1998, 
Aaronson_2004}. More specifically, a 
stabilizer state $\ket{\psi}$ on $L$ qubits can be uniquely 
defined by $L$ operators ${\cal P}_i$ according to 
${\cal P}_i \ket{\psi} = \ket{\psi}$, where the ${\cal P}_i$ are $L$-qubit 
Pauli strings. Instead of keeping track of the time evolution of the quantum 
state directly, $\ket{\psi} \to U(t)\ket{\psi}$, it is then useful to 
consider the evolution of the operators, ${\cal P}_i \to U(t){\cal 
P}_i U^\dagger(t)$. The latter can be done efficiently if 
$U(t)$ consists solely of Clifford gates which map single Pauli 
strings to other single Pauli strings, cf.\ Eq.\ \eqref{Eq::Cliff1}, in 
stark contrast to more general quantum evolution which 
would yield superpositions of multiple Pauli strings. More specifically, the 
$L$ stabilizers ${\cal P}_i$ can be encoded in a $L\times 2L$ binary matrix, 
also called the {\it stabilizer tableau}, the values of which are updated 
suitably upon the application of a Clifford gate.  As a consequence, the 
time and memory requirements scale only polynomially with the number of qubits, 
see Ref.\ \cite{Gottesman_1998} for further details. In fact, time-evolving the 
entire 
stabilizer tableau of a state $\ket{\psi}$ will here be necessary only for calculating the entanglement entropy in Fig.\ \ref{Fig::Entang}, whereas the analysis of 
operator spreading (Figs.\ \ref{fig:lightcone} and \ref{Fig::OpSpread}) and of 
the spectral form factor below (Fig.\ \ref{Fig::SFF}) requires the evolution 
of single (or, in case of the SFF, a specific set of) operators.   
Eventually, we note that there exist various 
efficient algorithms to generate random elements of the Clifford group 
\cite{Koenig_2014, Bravyi_2021}. 
We here follow the approach of Ref.\ \cite{Berg2020}, which we use to sample 
uniformly 
from ${\cal C}_2$ (${\cal C}_4$) in our simulations of 
1D (2D) circuit geometries. Except for Fig.\ \ref{fig:lightcone}, we 
then present results that are averaged over a sufficient number of random 
circuit realizations.    
\begin{figure}[tb]
 \centering
 \includegraphics[width=\columnwidth]{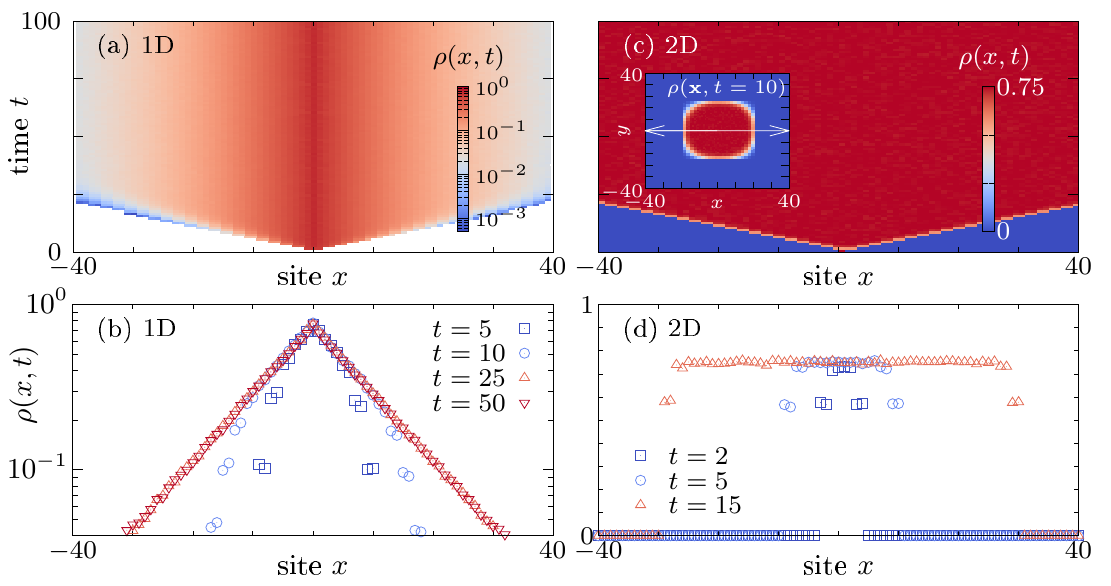}
 \caption{{\bf Operator spreading.} Circuit-averaged 
$\rho(x,t)$ in 1D [panels (a),(b)] and 2D [panels 
(c),(d)]. Data is obtained by time evolving an 
initially local operator $\sigma_x$ located at the origin, and averaging over 
$\sim 10^4$ random circuit realizations. Panels (b) and (d) show 
cuts of the data in (a) and (c) at fixed times $t$. For the better 
visualization of the 2D data, panel (c) shows $\rho(x,t)$ along a cut with $y = 
0$ [as indicated by the white arrow in the inset of panel (c), which depicts 
the full 2D data at $t = 10$]. While in 1D, the operator becomes 
exponentially localized, with $\rho(x,t) \propto e^{-\frac{x}{\mu}}$ and $ \mu \approx 10.4 $, the operator grows with light speed in 2D with a 
maximally scrambled light-cone interior as indicated by $\rho(x,t) \approx 
3/4$.}
 \label{Fig::OpSpread}
\end{figure}

\subsubsection{Operator spreading}

Figure \ref{fig:lightcone}~(a) demonstrates the occurrence 
of localization in 1D Floquet Clifford circuits. Specifically, considering  ${\cal O}(0) = \unity\otimes 
\cdots\otimes  \unity\otimes  \sigma_x\otimes  \unity\otimes  \cdots\otimes 
 \unity$, that differs from the identity only on one site, Fig.\ \ref{fig:lightcone} shows its time-evolution 
${\cal O}(t) = U(t) {\cal O}(0) U^\dagger(t)$ for a single 
random circuit realization. While the operator grows 
at light speed at short times, this growth eventually stops in 1D due to 
the emergence of left and right-blocking walls, 
leading to the Anderson-type localization behavior proven in Ref.\ 
\cite{Farshi_2022}. In contrast, this kind of localization is absent 
in the 2D circuit [Fig.\ 
\ref{fig:lightcone}~(b)], where we find that at least some parts of the 
time-evolved operator grow with light speed, that is, they act nontrivially on 
the light-cone boundary at all times.

In order to study this operator spreading in more 
detail, let ${\cal O}_x(t)$ denote the local Pauli matrix at the $x$-th 
position of the time-evolved operator. (For example, given a two-qubit 
system and ${\cal O}(t) = \sigma_x \otimes \sigma_z$, we might have ${\cal O}_1(t) = 
\sigma_x$ and ${\cal O}_2(t) = \sigma_z$.) We then define the quantity 
$\rho(x,t)$, with
\begin{equation}\label{def:rho}
 \rho(x,t) = \begin{cases}
              1\ ,\ &\text{if}\ {\cal O}_x(t) = \sigma_{x}, \sigma_y, \sigma_z 
\\
              0\ ,\ &\text{if}\ {\cal O}_x(t) = \unity
             \end{cases}\ , 
\end{equation}
In Fig.\ \ref{Fig::OpSpread}, we show the {\it circuit-averaged} dynamics of 
$\rho(x,t)$ for 1D and 
2D circuits. (Note that 
we here refrain from introducing another symbol for the average.)
On one hand, in the case of 1D circuits [Fig.\ \ref{Fig::OpSpread}~(a),(b)], we 
find that most of the 
operator's support remains close to the origin. More specifically, we find that 
the 
operator spreading is exponentially localized,  with $\rho(x,t) \propto e^{-\frac{x}{\mu}}$, $ \mu \approx 10.4 $, and becomes essentially 
stationary at long times, as can be seen by comparing the cuts of $\rho(x,t)$ 
at times $t = 25$ and $t = 50$ in Fig.\ \ref{Fig::OpSpread}~(b).

On the other hand, the situation is clearly different in 2D 
[Fig.\ \ref{Fig::OpSpread}~(c),(d)]. In this case, we find 
that $\rho(x,t)$ is essentially featureless inside the light cone with 
$\rho(x,t) \approx 3/4$, with a sharp drop to $\rho(x,t) \to 0$ outside the 
light cone. Note that for better visualization, the data in Fig.\ 
\ref{Fig::OpSpread}~(c) is shown for a cut at $y = 0$ of the original 2D 
data, cf.\ inset in 
\ref{Fig::OpSpread}~(c). Importantly, the circuit-averaged 
value $\rho(x,t) \approx 3/4$ indicates maximally scrambling dynamics. Namely, given 
the definition \eqref{def:rho} of $\rho(x,t)$ above, this value indicates 
that the matrices $\sigma_x,\sigma_y,\sigma_z$, and $\unity$ locally occur all
with 
equal probability.

It is in principle conceivable that there 
exist rare instances of random gate configurations such that the operator 
spreading becomes localized also in 2D. In practice, however, we do not 
observe any such localized instances and, in particular, there are no
notable signatures in the circuit-averaged behavior of $\rho(x,t)$ in Fig.\ 
\ref{Fig::OpSpread}~(c) and (d).

\subsubsection{Entanglement dynamics}

To proceed, we now turn to the buildup of entanglement  
$S(t)$ between two parts $A$ and $B$ of the system, starting from an initially 
unentangled product state $\ket{\uparrow}^{\otimes L}$ stabilized by 
$\sigma_z$ everywhere. Within the stabilizer formalism, $S(t)$ can be 
efficiently simulated based on the collective evolution of all $L$ operators 
${\cal P}_i$ that define $\ket{\psi}$ \cite{Nahum_2017, Hamma_2005}. Note that 
the entanglement dynamics in 
Clifford circuits is somewhat special since, given the reduced density matrix 
$\rho_A(t) = \text{tr}_B\ket{\psi(t)}\! \!\bra{\psi(t)}$, $\rho_A(t)$ will exhibit a 
flat eigenvalue spectrum such that all its R\'enyi entropies are equivalent 
\cite{Fattal2004}. 
Nevertheless, Clifford circuits can support extensive amounts of entanglement 
and the dynamics of $S(t)$ are often found to mimic the properties of more 
generic quantum evolutions \cite{Li_2018, Richter_2022}.
\begin{figure}[tb]
 \centering
 \includegraphics[width=\columnwidth]{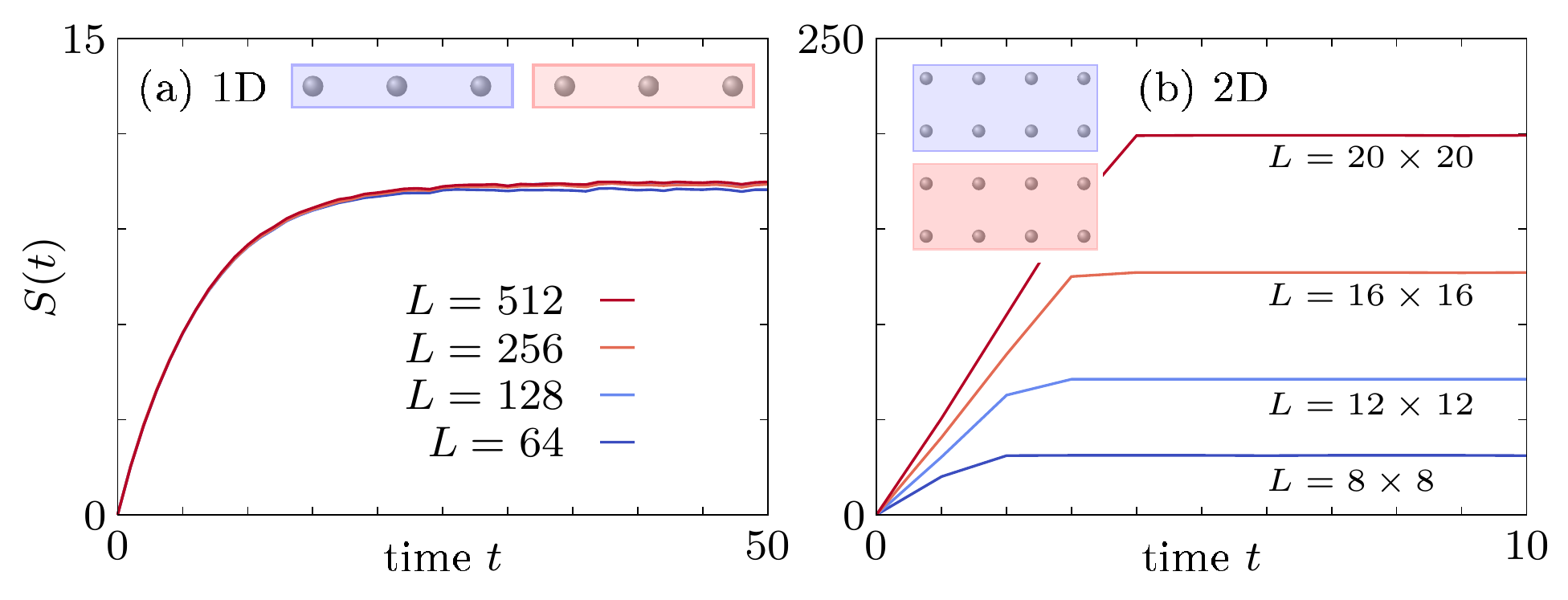}
 \caption{{\bf Entanglement growth.} Half-system 
(cf.\ sketches) entanglement entropy $S(t)$ under Floquet Clifford 
evolution for different system sizes $L$ in (a) one dimension and (b) two 
dimensions. In 1D, $S(t)$ approaches an almost system-size independent 
long-time value. In 2D, $S(t)$ saturates towards the maximally achievable value 
$L/2$. [Note the different scales of the y-axis in (a) and (b).] Data are 
averaged over $\sim 10^3$ random circuit realizations.}
 \label{Fig::Entang}
\end{figure}

In Fig.\ \ref{Fig::Entang}, we show $S(t)$ for half-system bipartitions in 1D 
and 2D Floquet Clifford circuits with various systems sizes $L$. The dynamics 
of 
$S(t)$ substantiate our previous observation of localization in 1D and 
ergodic behavior in 2D. Specifically, we find that $S(t)$ saturates towards a 
rather small and essentially $L$-independent long-time value $S(t\to\infty) 
\approx 10$ in 1D [Fig.\ \ref{Fig::Entang}~(a)]. This emphasizes that the 
localization length in 1D is distinctly 
shorter than the full system size and that a local operator can only explore a 
small part of the system, cf.\ Fig.\ \ref{fig:lightcone}~(a). 
In contrast, in 
2D [Fig.\ \ref{Fig::Entang}~(b)], we find that $S(t)$ exhibits a pronounced 
linear increase at short times, 
well-known from other random-circuit models and chaotic quantum systems.   
Moreover, at longer times, $S(t)$ saturates towards an extensive 
long-time value $S(t\to \infty) \approx L/2$, which is the maximally achievable 
value for a system of size $L$. This volume-law scaling of $S(t)$ highlights the 
absence of localization in 2D Floquet Clifford circuits.

\subsection{Spectral form factor}\label{Sec::SFF}

\subsubsection{Definition and general remarks}

To provide further insights into the nature of 
Floquet Clifford circuits we now turn to the SFF.
The SFF of an ensemble of unitaries $\mathcal U$ is defined as
\begin{equation}\label{Eq::SFF1}
 K (t) = \left \langle |\text{tr}\, U^t|^2 \right\rangle\ , 
\end{equation}
where $t$ takes integer values and the brackets $\langle \cdots \rangle$ denote 
average over all $U\in \mathcal U$. The SFF probes the statistical properties 
of 
 the spectrum of randomly sampled unitaries in $\mathcal U$. 
In particular, the Fourier transform 
\begin{align}
  \tilde K (\omega) = \frac 1 {\pi}
  \sum_{t=1}^\infty K(t)\, \cos(\omega t)\ ,
\end{align}
is the probability that a random $U\in \mathcal U$ has two eigenvalues  
separated by a distance $\omega$ (see \cite{Haake_book}).
The SFF also has an interpretation in terms of Poincar\'e recurrences, in fact 
$K(t)$ is a lower bound to the number of Pauli strings $\mathcal P$ which are 
mapped to themselves $U^{-t} \mathcal P U^t = \mathcal P$ after evolving for a 
time $t$.
More precisely, it was proven in \cite{Su_2020} that 
\begin{align}\label{eq:K rec}
  K(t) =
  2^{-L}\sum_{\mathcal P} \left\langle 
  \text{tr}\!\left[ \mathcal P U^{-t} \mathcal P U^t \right] 
  \right\rangle \ ,
\end{align}
where the sum is over all Pauli strings $\mathcal P$.
This identity holds for any unitary $U$, although only Clifford unitaries have  
the property that $U^{-t} \mathcal P U^t$ is equal or orthogonal to 
$\pm\mathcal 
P$. The possibility of this negative sign is responsible for the fact that 
\eqref{eq:K rec} is not necessarily equal to the number of recurrences, but a 
lower bound.

The SFF has been studied extensively to explore 
the onset of RMT behavior in the spectral properties 
of quantum systems \cite{M_ller_2004, Cotler2017, Gharibyan_2018, 
Prosen_2018_2, Prosen_2018_1}. 
In particular, it is proven in \cite{Haake_book} that, for unitaries drawn from the uniform distribution (Haar measure) over SU$(2^L)$, the SFF reads
\begin{align}\label{eq:rmt}
  K(t) = \left\{\begin{array}{ll}
    2^{2L} & \mbox{ if }  t=0 \\
    t & \mbox{ if } 1\leq t\leq 2^L \\
    2^L & \mbox{ if } t\geq 2^L \\
  \end{array}
  \right. .
\end{align}
Haar distributed unitaries go also under the name of CUE, circular unitary 
ensemble.
 Equation \eqref{eq:rmt} shows that the SFF increases as a linear ramp before  
reaching a plateau. Many-body chaotic systems are defined such that their SFF 
increase linearly in time after an initial dip. The time at which the linear 
ramp starts is called Thouless time, the time at which the SFF reaches a 
plateau 
is called Heisenberg time. A meaningful interpretation of these time scales 
requires rescaling, as discussed for example in \cite{Sierant_2020_2, 
Prosen_chal_2020}. Fig. \ref{Fig::SFF} shows that the SFF of the periodic 
Clifford circuit that we are considering manifests an exponential ramp before 
reaching a plateau. A similar behavior is shown by quasi-free fermions with 
chaotic single-particle dynamics \cite{Winer_2020, Liao_2020}. (By quasi-free 
fermions we mean those having a Hamiltonian quadratic in the creation and 
annihilation operators.)
Time in Fig.~\ref{Fig::SFF}  is not rescaled and the value at which a plateau  
is reached is ``evaluated'' by inspection therefore we call it plateau-time 
instead of Heisenberg time.

\begin{figure}[tb]
 \centering
\includegraphics[width=\columnwidth]{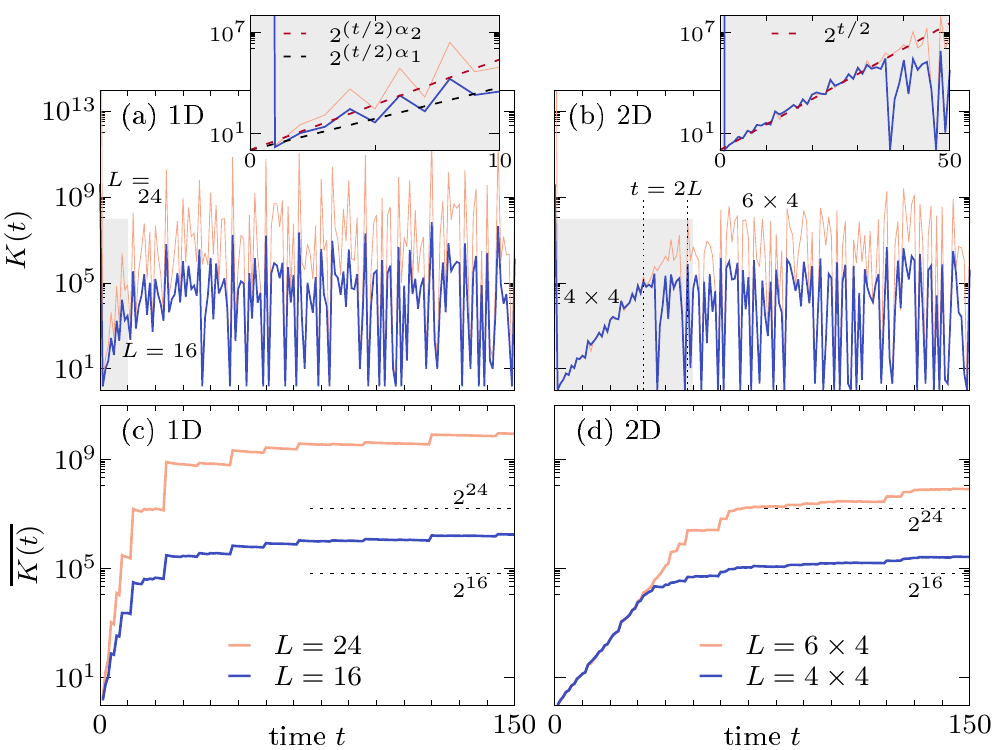}
\caption{{\bf Spectral form 
factor.} Circuit-averaged $K(t)$ of Floquet 
Clifford unitaries for different system sizes $L$ in (a) one dimension and (b) 
two dimensions. Data is averaged over $\sim 10^6$ random circuit 
realizations. The vertical dotted lines in (b) denote $t = 2L$, 
which approximately marks the end of the exponential ramp. The insets in (a) 
and (b) show the same data, but now for shorter times, cf.\ shaded areas in (a) 
and (b). In 2D, we find that the ramp is described by $K(t) \propto 2^{t/2}$, 
while in 1D $K(t) \propto 
2^{(t/2)\alpha}$ with $\alpha \approx 2.5$ for $L = 16$ and $\alpha \approx 
3.6$ for $L = 24$. In all cases, $K(t=0) = 4^L$. Panels (c) 
and (d) show the time-averaged SFF $\overline{K(t)}$ [Eq.\ \ref{Eq::OverKt}] for 
1D and 2D. The 
horizontal dashed lines indicate the Hilbert-space dimension~$2^L$.}
\label{Fig::SFF}
\end{figure}

Building on the Clifford formalism it 
is possible to compute Eq.\ \eqref{eq:K rec} efficiently \cite{Su_2020}, 
without 
 explicitly evaluating the exponentially many terms.
In order to do so, for any given Clifford unitary $U$, we define the group 
$H[U]$  of Pauli strings $\mathcal P$ that are invariant under the action of 
the 
unitary $U^\dagger \mathcal P U \pm \mathcal P$.  
Then, the method in \cite{Su_2020} consists of finding a set of generators of  
the group $H[U]$, denoted $\text{gen}H[U]$, and exploit the following identity
\begin{equation*}\label{Eq::SFF3}
 |\text{tr}U |^2 =
 \hspace{-0.1cm} \begin{cases}
         2^{|\text{gen}H[U]|}
         &\mbox{if } U^\dagger{\cal P}U =+{\cal P}
,\ \hspace{-0.1cm}\forall 
{\cal P} \in \text{gen}H[U] \\ 
0 &\text{otherwise}
        \end{cases} .
\end{equation*}
Then, this is done for $U= U_\text{per}^t$, with $ U_\text{per} $ as defined in \eqref{eq:Localization_2D_Absence_Depth2_Uperiod}, and several values of $t$.
Finally, one needs to average over many realizations of the circuit 
$U_\text{per}^t$.
The exponential dependence of $K(t)$ on the number of generators  
already hints at the fact that the SFF in Clifford circuits will
exhibit strong fluctuations.

\subsubsection{Numerical results}

We now present our numerical results for the SFF in Floquet 
Clifford unitaries. Note that while the stabilizer formalism in principle 
allows 
for simulations of $U(t)$ with complexity scaling polynomially with $L$, the 
system sizes in the following are smaller than one is typically used to from 
other examples of Clifford numerics. The main reason for this is that computing $K(t)$ requires  averaging over a rather high number of circuit 
realizations, which makes simulations of larger $L$ quite costly. Especially for 2D circuits, we find that extensive averaging is required to obtain converging results since the calculation of $K(t)$ appears to be dominated by rare circuit realizations that yield very large values of $|\text{tr}U_\text{per}^t|^2$.   

Figures \ref{Fig::SFF}~(a) and (b) show $K(t)$ for 1D and 2D 
circuits with two different system sizes $L = 16,24$. For the 1D circuit, we 
find that 
$K(t)$ exhibits a steep exponential increase at early times, but quickly 
crosses over to a plateau-like behavior. 
The time at which this crossover takes place
seems independent of $L$.
Strikingly, however, this ``plateau'' 
is dominated by major fluctuations of the SFF.  
Note that (most of) these fluctuations are actual features of the 
dynamics of $K(t)$ and cannot be reduced further by circuit-averaging. In 
fact, the data in Fig.~\ref{Fig::SFF} are already averaged over a substantial 
number ($\sim 10^6$) of random realizations. 
These fluctuations follow from the fact that the Clifford group has finite  
cardinality $|\mathcal C_L|$, which implies that the length $t$ of the orbits 
(i.e. the smallest $t$ such that $U^{-t} \mathcal P U^t = \mathcal P$) are 
divisors of $|\mathcal C_L|$. 
To see this we recall Lagrange's theorem: the order of a subgroup divides the 
order of the group. In particular, the smallest $r$ such that $U^r = \unity$ 
divides $|\mathcal C_L|$.
Next, let us show that $t$ is divisor of $r$ (and consequently divisor of 
$|\mathcal C_L|$). If we assume the opposite then $r=nt+m$ with $0\leq m<t$, 
which implies $\mathcal P = U^{-r} \mathcal P U^r = U^{-m} \mathcal P U^m$, 
which contradicts the fact that $t$ is minimal.

The behavior of $K(t)$ changes notably in the case of 2D 
circuits [Fig.\ \ref{Fig::SFF}~(b)]. Similar to 1D, we again find an 
exponential increase of $K(t)$ at early times. However, in contrast to 1D, 
this ramp is cleaner and persists for a longer time in 2D. Specifically, 
comparing $K(t)$ for different $L$, we find that the end of the ramp occurs at 
$t_\text{H} \approx 2L$, which corresponds to the phase-space dimension of the 
Clifford 
dynamics, cf.\ Sec.\ \ref{Sec::GenDyn}. For times $t > 2L$, the SFF displays
strong fluctuations around the plateau value. We 
interpret this behavior of $K(t)$ as a signature of ergodic dynamics in phase 
space, analogous 
to quasi-free fermions, as discussed above.

Let us explain why in the 2D case the plateau-time grows with the system's size as $t_\text{H} \approx 2L$ while in the 1D case $t_\text{H}$ is independent of 
$L$.
Equation \eqref{eq:K rec} tells us that $K(t)$ reaches the plateau when 
evolving Pauli strings experience recurrences. 
In the 2D case there is no localization, hence the evolution of Pauli strings is 
not restricted and can explore all the phase space. 
Recurrences reach a maximum at a time equal to the phase space dimension $2L$  
(analogously, in RMT recurrences reach a maximum at a time equal to the Hilbert 
space dimension \eqref{eq:rmt}).
In the 1D case the evolution of local Pauli operators is restricted to the  
localized regions, whose size depends on $U_\text{per}$ but not on the system's 
size $L$. The evolution of non-local operators is also restricted because they 
can be decomposed into local operators with restricted dynamics.
This follows from the fact that the evolution of a product of Pauli strings is 
equal to the product of the individual evolutions  
\begin{align}
  U^\dagger(t) [\mathcal P \mathcal P'] U(t) = [U^\dagger(t) \mathcal P U(t)] 
[U^\dagger(t)\mathcal P'U(t)]\ .
\end{align} 

Now, let us discuss the long-time behaviour of $K(t)$. In order to 
smooth the large fluctuations of $K(t)$, 
Fig.\ \ref{Fig::SFF}~(c) and (d) show the time-averaged SFF 
\begin{equation}\label{Eq::OverKt}
\overline{K(t)} = \frac{1}{t}\sum_{t'=0}^t K(t')\ . 
\end{equation}
The data of $\overline{K(t)}$ in Fig.\ \ref{Fig::SFF}~(c) emphasize the fact 
that the initial ramp of $K(t)$ in 1D becomes steeper with increasing $L$. 
The large-$t$ value of $\overline{K(t)}$ quantifies the degeneracies in 
$U_\text{per}$. This follows from 
\begin{align}\nonumber 
  \overline{K(\infty)} &= \lim_{t\to \infty} 
  \frac{1}{t} \sum_{t'=0}^t 
  \left\langle \sum_{i,j=1}^{2^L} e^{i(E_i-E_j)t'} \right\rangle 
  \\ \nonumber &= 
  \left\langle \sum_{i,j=1}^{2^L} \delta(E_i,E_j) \right\rangle 
  \\ &= \left\langle\sum_E g_E^2 \right\rangle \geq 2^L
  \ , 
\end{align}
where $E$ runs over all the quasi-energies of $U_\text{per}$ and $g_E$  is the 
degeneracy of $E$ (see \cite{Su_2020}). Note that when $g_E=1$ for all $E$ we 
have $\sum_E g_E^2 = 2^L$.
In Fig.\ \ref{Fig::SFF}~(c) we find that the  long-time value of 
$\overline{K(t)}$
is notably higher than the Hilbert-space dimension $2^L$, signaling the  
presence of degeneracy. In contrast, the RMT behavior of \eqref{eq:rmt}) 
implies 
that there are no degeneracies. 
An enhanced long-time value of $\overline{K(t)}$ can also be seen in 2D,  Fig.\ 
\ref{Fig::SFF}~(d), albeit less extreme in this case.

\section{Proof of theorem \ref{thm:Localization_2D_Absence_AlternativeModel_AbsenceResult}}
\label{sec:proof}

Let us consider the time evolution of a local operator acting on a single site 
at time $t = 0$. The support of this operator at later times 
($t > 0$) is 
contained in the light-cone represented in Fig.\ 
\ref{fig:Localization_2D_Absence_AlternativeModel_BoundaryGrowth}. 
But in general, an operator may not have full support in the light-cone 
(i.e.~it acts as the identity in some events).
In subsection \ref{sec:directed graph} we describe the boundary of the 
light-cone with a directed graph, and in subsection \ref{sec:rnd directed graph} 
we describe the operator growth with a random directed graph.

Our goal is to lower bound the probability that the evolution of a local 
operator has support at the boundary at all times (i.e.~non-identity in at least 
one site of the boundary). A crucial observation to address this question is 
that the state of the operator at the boundary at time $t$ depends only on the 
state of the operator at the boundary at time $t-1/2$, and does not depend on 
the state at the bulk of the light-cone. Importantly, all unitary gates at the 
boundary at time $t$ are statistically independent from the gates inside the 
light-cone at previous times. This fact makes this problem mathematically 
tractable.

\subsection{Directed-graph representation}
\label{sec:directed graph} 

The growth of the light-cone boundary can be represented by a directed graph 
$G$ constructed in the following way. 
Each vertex of $G$ represents a four-qubit unitary and each arrow represents 
a qubit that belongs to the boundary at a particular time $t$. 

\bigskip\noindent
\textbf{Construction of $G$:}
\begin{enumerate}
  \item Add one vertex $\bullet$ to represent the four-qubit unitary acting on 
  the initial site at $t=1/2$.
  \item Add one outwards arrow $\nearrow$ for each of the qubits where the 
  first unitary (potentially) propagates the signal at
$t=1/2$.
  \item Repeat the following steps starting at $t=1/2$:
  \begin{enumerate}
    \item At the end of each outwards arrow $\nearrow$ at time $t$, add the 
    vertex $\bullet$ corresponding to the unitary acting at $t+1/2$ on the 
qubit associated to the arrow $\nearrow$. (If two arrows represent qubits acted 
on by the same unitary at $t+1/2$, then these two arrows point to the same 
vertex.)
    \item From each vertex $\bullet$ corresponding to a unitary acting at  
$t+1/2$, add one outwards arrow for each of the qubits being acted on by the 
unitary $\bullet$ and belonging to the boundary at $t+1/2$.
  \end{enumerate}
  \item Update $t\mapsto t+1/2$ and repeat (a) and (b) up to infinity.
\end{enumerate}

\vspace{4mm}
\noindent

The construction of $G$ is detailed in Fig.\  
\ref{fig:Localization_2D_Absence_AlternativeModel_BoundaryGrowth}, and Fig.\ 
\ref{fig:Localization_2D_Absence_AlternativeModel_DirectedGraph} shows $G$ up to 
construction stage $t=3/2$.
Note that the infinite graph $G$ has the following property: unitaries which 
act on the qubits that are on the corners of the (rectangular) boundary are 
represented by a vertex with one inwards and three outwards arrows. 
Similarly, unitaries which act on qubits on edges of the boundary have two 
inwards and two outwards arrows.
\begin{figure}[tb]
\center
\includegraphics[width=0.85\columnwidth]{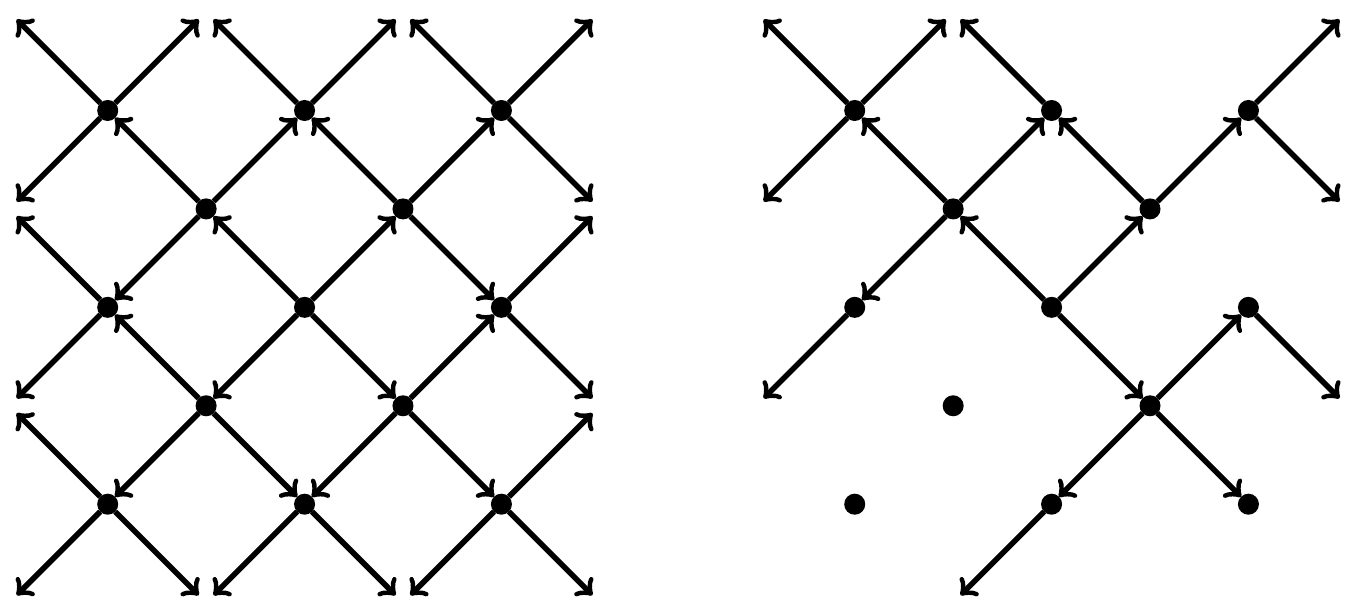}
\caption{\textbf{Directed graph at $t=3/2$.} In the left we have the directed 
graph $G$ while in the right we have an instance of the random directed graph 
$G_{\r}$.}
\label{fig:Localization_2D_Absence_AlternativeModel_DirectedGraph}
\end{figure}

\subsection{Random directed graph}
\label{sec:rnd directed graph}

In this section we represent the random growth of an initially local operator 
evolving according to the random dynamics 
\eqref{eq:Localization_2D_Absence_Depth2_Uperiod}, by a random directed graph 
$G_\r$ defined below. This graph is constructed by allowing the arrows of $G$ to 
be present or absent with certain probabilities. These probabilities are 
obtained from the statistical properties of the four-qubit gates $W_{(x,y)}$.
The property of light-speed operator growth up to time $t$ happens when there 
is a path in $G_\r$ starting at the central vertex, following the directions of 
the arrows, and having length $2t$. In what follows we bound the probability 
that $G_\r$ contains, at least, one such path.

There is, however, one minor point to consider.
In principle, if a vertex in $G_\r$ has no inwards arrows (i.e. the evolved 
operator has no support on the sites on which the gate associated to the vertex 
acts) then there should not be any outwards arrow also. 
However, since we are considering directed paths from the origin, this is not 
a problem. 
That is to say, if there are no inwards arrows to a vertex, the presence or 
absence of outwards arrows is irrelevant, since a directed path from the origin 
cannot use this vertex anyway. 
Therefore, we consider the presence of outwards arrows in a vertex of $G_\r$ 
to be statistically independent from the presence of inwards arrows, and 
independent of the presence of arrows in any other vertex. 
Still, the outwards arrows of a particular vertex are not independent random 
variables.
Each outwards arrow $i$ in a vertex has an associated random variable $x_i$ 
which takes the value $x_i=1$ if the arrow is present and $x_i=0$ if it is 
absent.
The following lemma provides the joint probability distribution of the 
variables $x_i$.

\begin{Lemma}
\label{lem:Localization_2D_Absence_Depth2_VertexProbDist}
For the vertex with four outwards arrows $x_1, x_2, x_3, x_4$, the probability 
for the presence of each arrow is
\begin{equation}
  \label{eq:Exact4outArrows}
  P( x_1 , x_2 , x_3 , x_4 ) = 
  \frac{1}{4^4-1}
  \begin{cases}
    0\   &\text{if $x_i = 0\ \forall\ i$} 
    \\
    3^{\sum_i x_i}\   &\text{otherwise} 
 \end{cases} \ , 
\end{equation}
where $x_i = 1$ indicates that the $i^{\text{th}}$ arrow is present, and 
$x_i = 0$ that it is absent.
For the vertices with three outwards arrows the probability distribution is
\begin{equation}
\label{eq:Exact3outArrows}
  P( x_1 , x_2, x_3 ) = \frac{1}{4^4-1}
  \begin{cases}
    3\  &\text{if $x_i = 0\ \forall\ i$} 
    \\
    4 \times 3^{\sum_i x_i} 
    &\text{otherwise}
  \end{cases} \ .
\end{equation}
For the vertices with two outwards arrows the distribution is
\begin{equation}
\label{eq:Exact2outArrows}
  P( x_1, x_2 ) = \frac{1}{4^4-1}
  \begin{cases}
    4^2-1 &\text{if $x_1 = x_2  = 0$}\  \\
    4^2 \times 3^{x_{1} + x_{2}}  &\text{otherwise }
  \end{cases}\ . 
\end{equation}
\end{Lemma}
\begin{figure}[tb]
\center
\includegraphics[width=0.45\columnwidth]{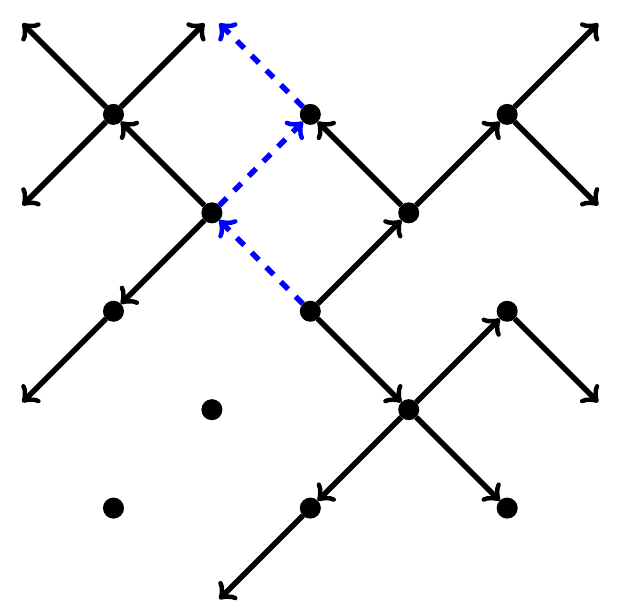}
\caption{This figure 
shows a $3$-path in an instance of the random directed graph $G_\r$.}
\label{fig:Localization_2D_Absence_AlternativeModel_DirectedGraph_WithLabels} 
\end{figure}

\begin{proof}
Let $A$ be a four-qubit Pauli operator different than the identity, and let 
$W$ be a four-qubit random Clifford unitary.
Let us consider the random  four-qubit Pauli operator 
$WAW^\dagger = \lambda B = \lambda B_1 \otimes B_2 \otimes B_3 \otimes B_4$ and 
ignore the phase $\lambda$. Lemma 3 from \cite{Farshi_2022} tells us that $B$ is 
uniformly distributed over the $4^4-1$ combinations of $B_1, B_2, B_3, B_4 \in 
\{\unity, \sigma_x, \sigma_y, \sigma_z\}$ different than $B_1= B_2= B_3=B_4 = 
\unity$.
The factor $3^{x_i}$ follows from the fact that the value $x_i=0$ stands for 
the one case $B_i =\unity$, while the value $x_i=1$ stands for the three cases 
$B_i \in \{\sigma_x, \sigma_y, \sigma_z\}$. This proves formulas 
(\ref{eq:Exact4outArrows}, \ref{eq:Exact3outArrows}, \ref{eq:Exact2outArrows}).
\end{proof}

\begin{Definition}
  The random directed graph $G_\r$ can be sampled by starting from $G$ and, at 
  each vertex, remove the outwards arrows according to the probability 
distributions (\ref{eq:Exact4outArrows}, \ref{eq:Exact3outArrows}, 
\ref{eq:Exact2outArrows}).
\end{Definition}

\begin{Definition}
An $l$-path in $G_\r$ is a sequence of $l$ consecutive arrows starting at the 
central vertex and following the directions of the arrows (see Fig.\  
\ref{fig:Localization_2D_Absence_AlternativeModel_DirectedGraph_WithLabels}).
\end{Definition}

\noindent
In the lemma below we show that in order to upper-bound the probability that 
$G_\r$ has no $l$-path, it is enough to analyse its lower quadrant, as 
represented in Fig.\ \ref{fig:Localization_2D_Absence_DigraphQuadrant_Depth2}.

\begin{Definition}
The directed graph $G^\triangle$ is the lower quadrant of $G$, and the random 
directed graph $G^\triangle_\r$ is the lower quadrant of $G_\r$ (see Fig.\ 
\ref{fig:Localization_2D_Absence_DigraphQuadrant_Depth2}).
\end{Definition}

\begin{Lemma}
\label{lem:Localization_2D_Absence_Depth4_UpperBoundQuadrant}
The probability that $G_\r$ has no $l$-path is upper-bounded by the probability 
that $G_\r^\triangle$ has no $l$-path, that is
\begin{equation}
 \pr \{ G_\r \text{ has no } l\text{-path} \}  \leq  \pr \{G_\r^\triangle 
 \text{ has no } l\text{-path} \} \ .
\end{equation}
\end{Lemma}

\begin{proof}
By definition, if $G_\r^{\triangle}$ has an $l$-path then $G_\r$ has an 
$l$-path too; therefore if $G_\r$ has no $l$-path then $G_\r^{\triangle}$ has no 
$l$-path. Using this, the bound follows.
\end{proof}
\begin{figure}[tb]
\center
\includegraphics[width=0.9\columnwidth]{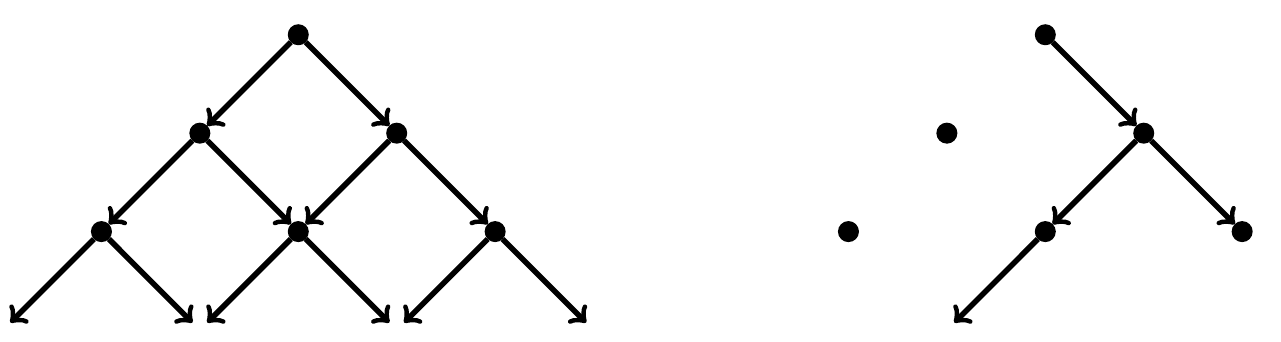}
\caption{In the left we have the graph 
$G^\triangle$, and in the right we have an instance of the random directed 
graph 
$G^{\triangle}_{\ir}$, both at time $t=3/2$.}
\label{fig:Localization_2D_Absence_DigraphQuadrant_Depth2}
\end{figure}

\subsection{Random graph with statistically-independent arrows}

In order to simplify the analysis, we define a new graph where the probability 
for the presence of each arrow is independent of the presence of other arrows. 
The next lemma shows that it is sufficient to analyze this new graph.

\begin{Definition}
  The random directed graph $G_\ir^{\triangle}$ is sampled by taking 
  $G^{\triangle}$ and independently removing each arrow with probability 
  \begin{align}\label{def:q}
    q := \pr\{x=0\}
    = \frac 1 2 - \frac 1 2\sqrt{\frac{21}{85}} \gtrapprox 0.25\ .
  \end{align}
\end{Definition}

\noindent
The subscript in $G_\ir^{\triangle}$ stands ``independent random arrows". 
The meaning of the value of $ q $ given above will become clear in the proof of 
lemma \ref{lem:Localization_2D_Absence_Depth4_IndependentArrows}. 
As mentioned above we have the following.

\begin{Lemma}
\label{lem:Localization_2D_Absence_Depth4_IndependentArrows}
The probability that $G_\r^\triangle$ has no $l$-path is upper bounded by the 
probability that $G_\ir^{\triangle}$ has no $l$-path:
\begin{equation}
\pr \{ G^\triangle_\r \text{ has no } l\text{-path} \} \leq \pr 
\{ G^{\triangle}_\ir \text{ has no } l\text{-path} \} \ . 
\end{equation}
\end{Lemma}

\begin{proof}
First, note that the probability distribution of outwards arrows from different 
vertices are independent, and hence we focus only on a single vertex with two 
outwards arrows. 
Second, construct a matrix with the probability distribution $P(x_1,x_2)$ for 
the presence of outwards arrows \eqref{eq:Exact2outArrows} as
\begin{equation}
  \begin{pmatrix}
    P(0,0) & P(0,1) \\
    P(1,0) & P(1,1)
  \end{pmatrix}   = \frac 1 {255}
  \begin{pmatrix}
    15 & 16\times 3 \\
    16\times 3 & 16\times 9
  \end{pmatrix} \ . 
\end{equation}
Third, we increase $\pr\{x_1=x_2=0\} \to P(0,0)+\epsilon$ and decrease
$\pr\{x_1=x_2=1\} \to P(1,1)-\epsilon$ until the distribution becomes of product 
form
\begin{align} \label{product_form}
  \begin{pmatrix}
    P(0,0)+\epsilon & P(0,1) \\
    P(1,0) & P(1,1)-\epsilon
  \end{pmatrix}  & = 
  \begin{pmatrix}
    \frac 1 {17} +\epsilon & \frac {16}{85} \\
    \frac {16}{85} & \frac {48}{85} -\epsilon
  \end{pmatrix} \\ &=
  \begin{pmatrix} 
    q^2 & q(1-q) \\
    q(1-q) & (1-q)^2
  \end{pmatrix} \ .
\end{align}
We want to find the minimum $\epsilon>0$ such that there is $q\in [0,1]$ 
satisfying the above equality.
Clearly, this transformation cannot increase the likelihood of an $l$-path.
The above matrix is of product form when its determinant is zero, in fact the 
second column of the last matrix in Eq. \eqref{product_form} is obtained from 
the first one multiplying by $ \frac{1-q}{q} $.
\begin{align}
  \det\begin{pmatrix}
    \frac 1 {17} +\epsilon & \frac {16}{85} \\
    \frac {16}{85} & \frac {48}{85} -\epsilon
  \end{pmatrix} =
  0\ .
\end{align}
This equation has smallest positive solution $\epsilon = \frac{1}{170} 
(43 - \sqrt{1785}) \approx 0.0044$, which implies
\begin{equation}
  q = \frac 1 {17} +\epsilon + \frac {16}{85} = \frac 1 2 - 
  \frac 1 2\sqrt{\frac{21}{85}} \ .
\end{equation}
This proves the statement of the lemma.
\end{proof}

\subsection{Dual graph}


\begin{Definition}
  The graph $G_\ir^{\triangle *}$ dual to $G_\ir^{\triangle}$ has vertices 
  located at the faces of $G_\ir^{\triangle}$. The presence of a (non-directed) 
edge in $G_\ir^{\triangle *}$ corresponds to the absence of the arrow it 
intersects in $G_\ir^{\triangle}$. 
  Therefore, the probability $q^*$ that an edge is absent 
  in $G_\ir^{\triangle *}$ is equal to the probability that an arrow is present 
in $G_\ir^{\triangle}$, that is  $q^* = 1-q$, where $q$ is defined in 
\eqref{def:q}.
  Figure \ref{fig:dual_graph} displays the graph dual to that of Fig.\
  \ref{fig:Localization_2D_Absence_DigraphQuadrant_Depth2}. 
\end{Definition}

\begin{Definition}
  A $d$-wall is a set of $d$ consecutive edges which connects the left side 
  to the right side of $G_\ir^{\triangle *}$.
\end{Definition}
\begin{figure}[tb]
\center
\includegraphics[width=0.9\columnwidth]{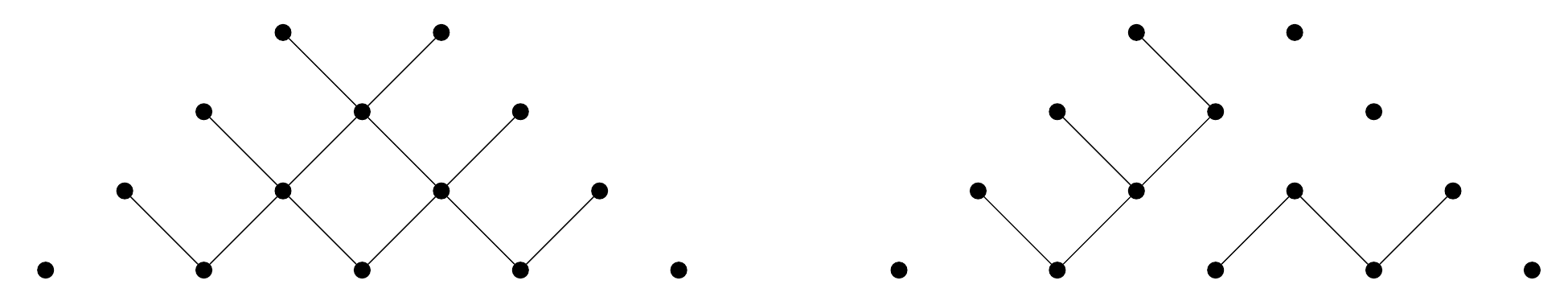}
\caption{\textbf{Dual graph.} On the 
left we have the dual graph $G^{\triangle *}$, and on the right we have the 
instance of the random graph $G_\ir^{\triangle *}$ that is dual to the instance 
of $G_\ir^{\triangle}$ shown in 
Fig.~\ref{fig:Localization_2D_Absence_DigraphQuadrant_Depth2}. 
}
\label{fig:dual_graph}
\end{figure}

\begin{Lemma}\label{lem:wall-path}
If $G_\ir^{\triangle *}$ contains a $d$-wall then $G_\ir^{\triangle}$ contains 
no $l$-path of length $l\geq d$.
\end{Lemma}

\begin{proof}
  This follows from the fact that a $d$-wall must start in one of the first 
  $d$ vertices of the left, and that $l$-paths always go downwards.
\end{proof}

\begin{Lemma}
\label{lem:Localization_2D_Absence_Depth4_UpperboundQuadrantPath}
The probability that $G_\ir^{\triangle *}$ has a $d$-wall for some $d$ satisfies
\begin{equation}
  \pr \{\exists\, d: G_\ir^{\triangle *} \text{ has } d\text{-wall} \} \leq 0.56  \ . 
\end{equation}
\end{Lemma}

\begin{proof}
  Let us start by upper-bounding the maximal number of $d$-walls $N_d$ that 
  $G_\ir^{\triangle *}$ can have. Consider a $d$-wall starting at a specific 
vertex position on the left side of $G_\ir^{\triangle *}$.
  For the choice of the first edge in the path there is only one possibility,  
for the choice of each of the following $d-2$ edges there are at most three 
possibilities, and for the final edge there is a single choice (again). Hence, 
we can upper-bound the number of $d$-walls starting at a specific vertex on the 
left side by $3^{d-2}$.
  It is worth noting that this upper-bound includes many paths which do not 
connect the left and right sides, and hence, do not actually form a $d$-wall. 
  The total number of vertices which can be the initial vertex (left side) of 
  a $d$-wall is $d-1$.
  Hence, we have $N_d \leq (d-1) 3^{d-2}$. In order to obtain a better bound 
  we note that, when the starting vertex of a $d$-wall is either the 
$1^{\text{st}}$ or $(d-1)^{\text{th}}$ from the top, then there is only one 
possible choice of $d$-wall. 
  Therefore, we obtain
\begin{equation}\label{bound:Nd}
  N_d \leq (d-3) 3^{d - 2} +2 \ .
\end{equation}
Direct counting gives the exact number of $d$-walls for small values of $d$, 
which is displayed in the following table.
\begin{equation}\label{eq:tableNd}
\begin{tabular}{|l|l|l|l|l|l|}
  \hline
  $d$ & 2 & 3 & 4 & 5 & 6 \\ \hline
  $N_d$ & 1 & 2 & 3 & 6 & 18 \\ \hline
\end{tabular}
\end{equation}
Next, the probability that $G_\ir^{\triangle *}$ contains a particular $d$-wall 
in $G^{\triangle *}$ is $(1-q^*)^d =q^d$. Therefore, the probability that 
$G_\ir^{\triangle *}$ has at least one $d$-wall satisfies
\begin{equation}
  \pr \{G_\ir^{\triangle *} \text{ has } d\text{-wall} \} \leq N_d\, q^d \ .
\end{equation}
Finally, the probability that $G_\ir^{\triangle *}$ contains a $d$-wall of any 
length $d$ satisfies 
\begin{align}
  \nonumber
  \pr \{&\exists\, d: G_\ir^{\triangle *} \text{ has } d\text{-wall} \} \\ 
  &\leq\ 
  \sum_{d=2}^\infty \pr \{G_\ir^{\triangle *} \text{ has } d\text{-wall} \}  
  \leq\ 
  \sum_{d=2}^\infty N_d\, q^d 
  \\ &\leq\ 
  \sum_{d=2}^6 N_d\, q^d +
  \sum_{d=7}^\infty \left[(d-3)3^{d-2}+2\right] q^d \\
 & \approx 0.56\ ,
\end{align}
where we have use table \eqref{eq:tableNd} and bound \eqref{bound:Nd}.
\end{proof}

Combining lemmas \ref{lem:Localization_2D_Absence_Depth4_UpperBoundQuadrant}, 
\ref{lem:Localization_2D_Absence_Depth4_IndependentArrows}, \ref{lem:wall-path} 
and  \ref{lem:Localization_2D_Absence_Depth4_UpperboundQuadrantPath} we can 
prove our main result.

\medskip\noindent
\textbf{Theorem 3.}
The probability that $G_\r$ has an $l$-path of infinite length $l=\infty$ 
satisfies
\begin{equation}\label{Eq::ProofFinal}
  \pr \{G_\r \text{ has } \infty\text{-path} \} \geq 0.44  \ . 
\end{equation}

To conclude this section, let us note that while this proof of 
theorem\ \ref{thm:Localization_2D_Absence_AlternativeModel_AbsenceResult} 
establishes the absence of localization in the 2D Floquet Clifford 
circuits \eqref{eq:Localization_2D_Absence_Depth2_Uperiod}, the lower bound 
given in Eq.\ \eqref{Eq::ProofFinal} is probably 
not tight for practical purposes. Specifically, in our simulations of  
2D circuits of {\it finite size} $L < \infty$, we find nontrivial operator 
support on the light-cone boundary for virtually all times and all random 
circuit 
realizations, such as in the example depicted in Fig.\ \ref{fig:lightcone}~(b).

\section{Comparing Clifford dynamics with quasi-free bosons and fermions}
\label{sec:int/chao}

In this section we argue that Clifford dynamics shares features with quasi-free 
systems, along with certain similarities with chaotic systems. Therefore, in order 
to elucidate the full landscape of quantum many-body phenomena, it is important to understand the 
properties of Clifford systems.

\subsection{General dynamics}
\label{Sec::GenDyn}

Similarly to quasi-free fermions, Clifford unitaries can be represented by 
symplectic matrices in a phase space of dimension exponentially smaller than 
the Hilbert space. This dimensional reduction allows for efficient simulation of 
the evolution of any local or Pauli operator with a classical computer 
\cite{Aaronson_2004}. 
The efficient simulability of Clifford circuits can also be 
understood with respect to the non-negativity of the associated Wigner 
function in phase space \cite{Gross_2006a, Mari_2012}. In particular, this 
non-negativity of the Wigner 
function of stabilizer states is analogous to the non-negative Wigner function 
of the Gaussian states corresponding to models with quadratic Hamiltonian 
\cite{Hudson_1974, Weedbrook_2012}. 
Furthermore, we have seen in Sec.\ \ref{sec:proof} that Clifford dynamics 
allows for a certain degree of analytical tractability, which is similar to 
other types of integrable models.

Unlike quasi-free systems, the Clifford phase space is a vector space over a 
finite field ($\mathds{Z}_2^{2N}$ is the phase space of $ N $ qubits) hence 
evolution cannot be continuous in time. That is, we can have Floquet-type but 
not Hamiltonian-type dynamics.
For the same reason, the evolution operator 
cannot be diagonalised into non-interacting ``modes". 
Related to this is the fact that some aspects of typical 
Clifford dynamics cannot be understood 
in term of free or weakly-interacting particles [see for instance the dynamics 
in Fig.\ \ref{fig:lightcone}~(a)].

Random Clifford unitaries resemble uniformly distributed (i.e.\ Haar) unitaries 
much more than random quasi-free unitaries do.
This can be made quantitative by using the notion of unitary design 
\cite{Eisert_2007}. 
On one hand, random quasi-free unitaries cannot even generate a 1-design, 
because their evolution operators commute with the number operator (bosons) or 
the parity operator (fermions). On the other hand, random Clifford unitaries 
generate a 3-design \cite{Webb_2016,Zhu_2017}, and almost a 4-design 
\cite{Zhu_2016}. 

As discussed in Sec.\ \ref{Sec::SFF}, the SFF of Clifford unitaries corresponds to that of quasi-free fermions with chaotic single-particle dynamics, as in the quadratic SYK model \cite{Winer_2020, Liao_2020}.

\subsection{Translation-invariant local dynamics}

In \cite{Zimboras_2020} the authors analyze a translation-invariant Clifford 
Floquet model in one spatial dimension. They prove that the system has no local 
or quasi-local integrals of motion.
More specifically, any operator that commutes with the evolution operator acts 
on an extensive number of sites with couplings among them that do not decay with 
the distance. This is very different to what happens with quasi-free systems, 
which have an extensive number of local conservation laws (see 
\cite{BR_2,Dierckx_2008}).

Unlike quasi-free systems, the Clifford model analyzed in \cite{Zimboras_2020} enjoys a 
strong form of eigenstate thermalization. That is, all eigenstates are maximally 
entangled in the sense that the reduced density matrix of any connected region 
is equal to the maximally-mixed state (in the thermodynamic limit). In other 
words, none of the eigenstates satisfy an entanglement area law.

\subsection{Disordered local dynamics}

The 1D analog of our model (analyzed in 
detail in \cite{Farshi_2022}) 
displays a strong form of localization produced by the emergence of left and 
right-blocking walls at random locations [see Fig.\ 
\ref{fig:lightcone}~(a)]. Until 
now, this strong form of localization, reminiscent of 
Anderson localization, 
has only been found in systems of free or weakly-interacting particles. In 
strongly-interacting systems the localization is in some sense weaker (many-body 
localization), since the width of the light-cones grows as the logarithm of 
time. Clifford dynamics appears to be some form of hybrid as 
it cannot be understood entirely in terms of free or 
weakly-interacting particles, but yet displays 
Anderson-type localization in 1D.
Remarkably, however, the similarity in their localization 
properties does not carry over to 2D. More specifically, while we have 
shown in this paper that localization is absent in 2D random Floquet Clifford 
circuits, quasi-free systems, such as non-interacting fermions, are 
well-known to localize in 2D lattices even in the presence of arbitrarily weak 
disorder \cite{Lee_1981}.

\section{Conclusions}
\label{sec:conclusions}

We have introduced a Floquet model comprised of random 
Clifford gates and proved as a main result that it does not localize 
in two spatial dimensions, despite having strong disorder.  
More precisely, we have proven the existence of operators that grow 
ballistically, and we have seen numerically that this holds for almost all 
operators. 
This result is notable because, as discussed in 
Sec.\ \ref{sec:int/chao}, this model shares certain features 
of other 
integrable models, and 
two-dimensional integrable systems with strong disorder are 
usually expected to localize, e.g., non-interacting lattice fermions in a 
random potential as originally considered in the case of Anderson localization. 
We have substantiated our analytical findings by 
numerically demonstrating the absence (presence)  of localization in 
2D (1D) Floquet Clifford circuits. Furthermore, we 
studied the spectral form factor of the Floquet unitary, which is a key 
quantity in the context of quantum chaos. To the best of our knowledge, our 
work is the first to study the SFF in 2D Clifford 
circuits, complementing the analysis of Ref.\ \cite{Su_2020}, which has focused 
on 1D. 
We have unveiled that the SFF behaves drastically different in 
2D. In particular, we observe a clean exponential ramp 
persisting up to a time that scales linearly with system size, which we 
interpret as a signature of ergodic dynamics in phase space.

It is worth noting that, since the definition of light-speed growth only 
concerns the boundary of the light-cone, our results also applies to the 
time-dependent (non-Floquet) version of the model, were new gates are randomly generated at 
each time step. 
For the same reason our results for the two-dimensional model extend to time-periodic circuits with time-period larger than two.
The difference between models with period equal to two and models with a larger period or time-dependent models  could manifest itself in the interior 
of the light cone. While the time-dependent model has completely ergodic 
dynamics, Fig.\ \ref{Fig::OpSpread}~(c) suggests that our 
time-periodic model is 
not far from it.  
While with our analytical methods we cannot 
address the interior of the light-cone, it might be interesting to extend our 
numerical analysis to probe whether the bulk of the light cone is indeed 
featureless or whether the Floquet circuit induces some structure on the 
Pauli strings that are sampled during the dynamics. Such potential differences 
compared with fully ergodic dynamics may for instance reflect themselves in the 
behavior 
of higher-order and non-local correlation functions that go beyond the local 
probe considered in Fig.\ \ref{Fig::OpSpread}.

In future work we would like to study the case where instead of sampling gates 
from the Clifford group we sample from the full unitary group SU(4). 
In the case of time-dependent circuits this has been well studied in references 
such as \cite{Nahum_2018, Khemani_2018a, Harrow_2018}.
Whereas the case of time-periodic quantum circuits in one spatial dimension 
has been studied in references such as 
\cite{Prosen_2018_1, Prosen_2018_2, Bertini_2019}. 
In this context, one promising approach to interpolate between 
Clifford dynamics and more generic Haar-random circuits is to consider 
circuits composed mainly of Clifford elements interspersed with a (low) density 
of non-Clifford gates, the latter acting as a seed of chaos which may 
enhance the ergodic 
aspects of the dynamics further \cite{Zhou_2019_2, Eisert_2020}, despite loosing classical simulability \cite{Leone_2021}. 
Such a procedure may also help 
to better understand the apparent differences as well as similarities of 
Clifford dynamics and other notions of integrability \cite{Winer_2020, 
Liao_2020}, a detailed exploration of 
which we believe to be an important direction of future work. 

The localization proven for dynamics of a period-two, one-dimensional QCA of Cliffords in \cite{Farshi_2022} and corroborated numerically in Fig. \ref{Fig::OpSpread}, as seen above, is known to disappear in the limit of a time-dependent circuit, as it also follows from \cite{Farshi_2022}. Numerical investigations performed in \cite{Taylor_2022} show that the one-dimensional circuit with period equal to four still localizes with the appearance of hard-walls that nevertheless are characterised by a larger localization length than the period-two case.  

Finally, we would like to comment on the connections between our methods and 
directed percolation theory \cite{Saberi_2015, Durrett_1984}. 
Both cases analyze the presence of infinitely long paths which start at the 
origin in random directed graphs. 
But in our case, and in that of general cellular automata, the 
arrows that emerge from the same vertex are not statistically independent, see 
Lemma \ref{lem:Localization_2D_Absence_Depth2_VertexProbDist}, while in the 
standard theory of directed percolation they are.

In conclusion, our work provides a new perspective on the 
possibility of using Clifford circuits to simulate certain novel nonequilibrium 
quantum phenomena. We expect that our work will inform future studies that aim 
to use Clifford circuits as 
starting points to understand 
more generic quantum dynamics. 

The rigorous 
understanding of localization and chaos in this solvable limit provides a basis 
for toy models to simulate non-equilibrium states in kinetically constrained 
models. The classical simulability of Clifford circuits can play a constructive 
role for benchmarking the performance of noisy intermediate-scale quantum 
computers for quantum simulation and distinguishing between classical and 
quantum effects in them.

\subsection*{Acknowledgments}
Tom Farshi acknowledges financial support by the Engineering and 
Physical Sciences Research Council (grant number EP/L015242/1 and EP/S005021/1) 
and is grateful to the Heilbronn Institute for Mathematical Research for 
support.
Lluis Masanes and Daniele Toniolo acknowledge 
financial support by the UK's Engineering and Physical Sciences Research Council 
(grant number EP/R012393/1). Daniele Toniolo also acknowledges support from UKRI 
grant
EP/R029075/1. 
Jonas Richter and Arijeet Pal acknowledge funding by the
European Research Council (ERC) under the European
Union’s Horizon 2020 research and innovation programme
(Grant agreement No. 853368). Jonas Richter 
also received funding from the European
Union's Horizon Europe programme
under the Marie Sk\l odowska-Curie grant agreement
No. 101060162.
    
\section*{Appendix}
\appendix

\section{Localization length in the 1D model} \label{App}
We present here a result that builds upon, and improves, theorem 25 of \cite{Farshi_2022}. Let us first informally restate this theorem: The periodic dynamics of a one-dimensional QCA of two-qubits gates of uniformly sampled Clifford unitaries, represented as the inbox of Fig.~\ref{fig:lightcone}, and Fig.~1 of \cite{Farshi_2022}, is such that the probability of appearance of a right  (or left) blocking wall is at least  $0.12$. If we consider only walls with penetration length equal to $1$, meaning that if for example a right-blocking wall is placed at $ x=0 $ then the support of an operator hitting that wall is allowed to flow only till $x=1$, then the probability of appearance of a wall is exactly equal to $0.12$. The proof of this statement rests on the use of the phase-space formalism for Clifford unitaries. The following corollary gives an upper bound on the average spread of the support of an operator initially non identity only at one point. This spread is the localization length $ \mu $ that has been given in definition  \ref{def_Anderson_type}. 
\begin{Corollary} \label{cor_loc_length}
 The periodic dynamics of a one-dimensional QCA of two-qubits gates of uniformly sampled Clifford unitaries is such that given an operator initially supported only on the site $ x=0$, the probability that a wall blocking its dynamics appears at a distance $ l $, and no other wall appears at $\{0,\dots, l-1\}$, is $ ce^{-\frac{l}{\mu}}$, with $ c \approx 0.07 $ and  $ \mu $ upper bounded by  $13.2$. 
\end{Corollary}
\begin{proof}
 Let us assume for simplicity that we consider a wall blocking propagation towards the right. According to the phase space formalism, as detailed in reference \cite{Farshi_2022}, see in particular section II E and section VII of \cite{Farshi_2022}, the absence of walls at neighboring positions, for example $ x=0 $ and $ x=1 $, are correlated events. In fact they satisfy the conditions: $ C_1C_0 \neq 0 $ and$\setminus$or $ C_1D_0A_1C_0 \neq 0 $ for the absence of a wall at $ x=0 $, and  $ C_2C_1 \neq 0 $ and$\setminus$or $ C_2D_1A_2C_1 \neq 0 $ for the absence of a wall at $ x=1 $. Denoting $ W_l $ the presence of a wall at $ x=l $, and $ \overline{W}_l $ the absence of a wall at $ x=l $, we see that $ \textrm{Prob} ( \overline{W}_1 \wedge \overline{W}_0 ) = \textrm{Prob} ( \overline{W}_1 | \overline{W}_0 ) \textrm{Prob} ( \overline{W}_0 )$.
 This leads to a Markov chain, in fact the absence of a wall at $ x=2 $ involves the matrix blocks $ \{C_3,D_2,A_3,C_2\} $, we see that none of them appears in the condition defining $\overline{W}_0 $. This means that the probability of having a wall at $ x=l $  and no other wall at $\{0,\dots, l-1\}$ is given by:
 \begin{align}
  &\textrm{Prob} \left( W_l \wedge \overline{W}_{l-1}  \wedge \dots \wedge \overline{W}_0 \right) = \nonumber \\
 &=\textrm{Prob} \left( W_l | \overline{W}_{l-1} \right) \textrm{Prob} \left( \overline{W}_1 | \overline{W}_0 \right)^{l-1} \textrm{Prob} \left( \overline{W}_0 \right) \nonumber \\
&= \left(1-\textrm{Prob} \left( \overline{W}_l | \overline{W}_{l-1} \right) \right) \textrm{Prob} \left( \overline{W}_1 | \overline{W}_0 \right)^{l-1} \textrm{Prob} \left( \overline{W}_0 \right) \nonumber \\
&= \left(1-\textrm{Prob} \left( \overline{W}_1 | \overline{W}_0 \right) \right) \textrm{Prob} \left( \overline{W}_1 | \overline{W}_0 \right)^{l-1} \textrm{Prob} \left( \overline{W}_0 \right) \nonumber \\
&= \frac{\left(1-\textrm{Prob} \left( \overline{W}_1 | \overline{W}_0 \right) \right) \textrm{Prob} \left( \overline{W}_0 \right) }{\textrm{Prob} \left( \overline{W}_1 | \overline{W}_0 \right)} \textrm{Prob} \left( \overline{W}_1 | \overline{W}_0 \right)^{l}  \label{l_wall}
 \end{align}
 With the aid of a computer software we have exactly evaluated $ \textrm{Prob} \left( \overline{W}_1 | \overline{W}_0 \right) = 0.927 $, and we recall from \cite{Farshi_2022} that $ \textrm{Prob} \left( \overline{W}_0 \right) = 0.88 $. Then equation \eqref{l_wall} implies that 
\begin{equation} \label{exp}
 \textrm{Prob} \left( W_l \wedge \overline{W}_{l-1}  \wedge \dots \wedge \overline{W}_0 \right) = c e^{-\frac{l}{\mu}}
\end{equation}
with
\begin{align}
 & c =  \frac{\left(1-\textrm{Prob} \left( \overline{W}_1 | \overline{W}_0 \right) \right) \textrm{Prob} \left( \overline{W}_0 \right) }{\textrm{Prob} \left( \overline{W}_1 | \overline{W}_0 \right)} \approx 0.07 \label{c} \\
 & \mu = -\frac{1}{\log \left( \textrm{Prob} \left( \overline{W}_1 | \overline{W}_0 \right) \right) } \approx 13.2 \label{mu}
\end{align}
This proof takes into account walls with penetration length equal to 1, described by the Fig. 4 and 5 of \cite{Farshi_2022}. Walls with larger penetration length are also allowed by the dynamics, taking them into account would lead to a shorter localization length, as confirmed by the numerics that leads to Fig. \ref{Fig::OpSpread}, therefore the value  $ \mu \approx 13.2 $ is an upper bound. 
\end{proof}
We comment about normalization of probability as a check that corollary \ref{cor_loc_length} is consistent. 
In the limit of an infinite system the probability of having no wall at all is vanishing. This follows from the proof of corollary \ref{cor_loc_length} where instead of evaluating $  \textrm{Prob} \left( W_l \wedge \overline{W}_{l-1}  \wedge \dots \wedge \overline{W}_0 \right) $ we consider $ \textrm{Prob} \left( \overline{W}_l \wedge \overline{W}_{l-1}  \wedge \dots \wedge \overline{W}_0 \right) $. This means that the probability of having at least one wall equals 1. The probability of having at least one wall is the sum of the probability of having a wall in $ x=0 $ and whatever else, plus the probability of having no wall in $ x=0 $ and a wall in $ x=1 $ and whatever else, plus the probability of having no wall in $ x=0 $ and no wall in $ x=1 $ and a wall in $ x=2 $ and whatever else, and so on. This means that:
\begin{equation} \label{norm}
 1 = \textrm{Prob} \left( W_0 \right) + \sum_{l=1}^\infty \textrm{Prob} \left( W_l \wedge \overline{W}_{l-1}  \wedge \dots \wedge \overline{W}_0 \right) 
\end{equation}
It possible to verify that equation \eqref{norm} is implied by equations \eqref{exp}, \eqref{c} and \eqref{mu}.
We then have the normalized probability distribution $ P(l) $ for the appearance of a wall at $ x=l $ and no prior wall:
\begin{align}
  P(l) = \left\{\begin{array}{ll}
    \textrm{Prob} \left( W_0 \right) = 0.12 & \mbox{ with }  l=0 \\
    c e^{-\frac{l}{\mu}} & \mbox{ with } l \ge 1 \\
  \end{array}
  \right.
\end{align}
This implies that the average position of a wall is given by 
\begin{equation}
 \langle l \rangle = \sum_{l=1}^\infty l c e^{-\frac{l}{\mu}} = \textrm{Prob} \left( \overline{W}_0 \right)  \mu + 
\mathcal{O}\left(1-\textrm{Prob} \left( \overline{W}_1 | \overline{W}_0 \right) \right)^2
\end{equation}
This gives the relationship among the average position of a wall and the localization length $ \mu $ in the limit of a large system.

\bibliographystyle{apsrev4-2}
\bibliography{bibliography}

\end{document}